\documentclass[runningheads,a4paper]{llncs}

\usepackage[T1]{fontenc}
\usepackage{stmaryrd}
\usepackage{algorithm}
\usepackage{algorithmic}
\usepackage{graphicx}
\usepackage{floatflt}
\usepackage{wrapfig}
\usepackage{pstricks}
\usepackage{rotating}
\usepackage{subfigure}
\usepackage{amsmath}
\usepackage{amssymb}
\usepackage{theorem}
\usepackage[active]{srcltx} 
\usepackage{amsmath,amssymb}
\usepackage{booktabs}
\usepackage{xcolor}
\usepackage{hyperref}
\usepackage{listings}
\usepackage{enumerate}
\usepackage[justification=RaggedRight, singlelinecheck=false]{caption}
\usepackage{marvosym}

\newcommand{\QED}{\hspace*{\fill}$\Box$}

\newcommand{\T}{{\cal T}}

\newcommand{\A}{{\cal A}}
\newcommand{\B}{{\cal B}}
\newcommand{\C}{{\cal C}}

\newcommand{\M}{{\cal M}}

\newcommand{\ignore}[1]{}

\newcommand{\ite}[3]{\text{ {\sf if} }#1\text{ {\sf then} }#2\text{  {\sf else} }#3}

\newenvironment{enumerate-} 
{\begin{enumerate}
    
   \setlength{\parskip}{-1ex}              
   \setlength{\itemsep}{1.5ex}             
}
{
 \end{enumerate}
}

\newtheorem{defi}{Definition}
\newtheorem{thm}{Theorem}
\newtheorem{lem}[thm]{Lemma}

\newtheorem{rem}{Remark}
\newtheorem{ex}{Example}

\begin{document}
\pagestyle{headings}
\title{On Constructing Most General Solutions for Parametric
  Constraints (Extended Preprint)}
 \titlerunning{On Constructing Most General Solutions for Parametric Constraints}
\author{Viorica
  Sofronie-Stokkermans\orcidID{0000-0002-8486-9955}}
\authorrunning{Viorica Sofronie-Stokkermans}
\institute{University of Koblenz, Koblenz, Germany \\
\email{sofronie@uni-koblenz.de}}
\maketitle

\begin{abstract}
Let $\T$ be a theory allowing a form of elimination of 
existential quantifiers (possibly for formulae in a certain class).
We analyze possibilities of 
constructing (most general) solutions w.r.t.\ $\T$ for 
formulae of the form $\exists x_1 \dots \exists x_n \phi(x_1, \dots, x_n, y_1, \dots, y_m)$, 
where $\phi$ is a quantifier-free conjunction of literals 
in the signature of $\T$, and the free variables $y_1, \dots, y_m$ 
are regarded as parameters. 
We show that in the presence of function symbols which describe ``{\sf if}-{\sf then}-{\sf
  else}''  constructions in certain models of $\T$, we can describe 
the most general solution of such formulae, thus generalizing 
results about the existence of most general unifiers in discriminator
varieties. We illustrate the ideas on examples. 
\end{abstract}

\section{Introduction}

Finding solutions for parametric constraints is important in many
applications. 
While state of the art SMT-systems can generate models for 
satisfiable formulae, they usually do not distinguish between 
parameters (constants/function symbols) and existentially 
quantified symbols. 
In this paper we analyze possibilities of 
constructing parametric solutions and even 
most general solutions of parametric constraints 
w.r.t.\ logical theories $\T$ allowing some form of quantifier 
elimination. 
The constraints we consider are formulae of the form 
\begin{equation}
\exists x_1, \dots, x_n \phi(x_1, \dots, x_n, y_1, \dots, y_m), \label{eq1}
\end{equation}
where $\phi$ is a conjunction of literals in the signature of $\T$.  
Our goals are two-fold: 
\begin{itemize}
\item[(i)] Find a necessary and sufficient condition on 
$y_1, \dots, y_m$, seen as parameters, under which 
$\exists x_1, \dots, x_n \phi(x_1,\dots, x_n, y_1, \dots, y_m)$ 
evaluates to true w.r.t.\
$\T$. 

\item[(ii)] Analyze possibilities of obtaining 
most general solutions for problems of form~(\ref{eq1}), i.e.\ 
substitutions from which all solutions can be retrieved as instances. 
\end{itemize}
We address (i) by showing that if $\T$ allows a form of quantifier elimination 
for a fragment to which the formula $\phi$ belongs, the 
quantifier free formula $\psi(y_1, \dots, y_m)$ 
equivalent to 
$\exists x_1, \dots, x_n \phi(x_1, \dots, x_n, y_1, \dots, y_m)$ 
w.r.t.\ $\T$ is such a necessary and sufficient condition.
For (ii), we show that in the presence of suitable ``{\sf if}-{\sf then}-{\sf else}'' constructions 
in certain 
models of $\T$ we can effectively 
describe the most general solution. 

Our results are inspired by the study of 
unification in Boolean algebras
\cite{Boole,Loewenheim,Rudeanu,Baader98}, 
in varieties generated by primal algebras \cite{Nipkow} 
and, more generally, in discriminator varieties \cite{Burris92}.
Consider the case of unification in Boolean
algebras i.e.\ problems of the form
$\exists x ((t_a \wedge x) \vee (t_b \wedge \neg x) \approx 0),$
where $t_a$ and $t_b$ are terms containing other variables, 
which can be considered to be parametric.

It is known \cite{Boole,Loewenheim,Rudeanu} that for every 
Boolean algebra $B$, if $a, b \in B$ then the equation 
$(a \wedge x) \vee (b \wedge \neg x) \approx 0$ has a solution $x$ in
$B$ iff $a \wedge b = 0$, 
and in this case the set of solutions in $B$ 
is the interval 
$S = [b, \neg a]$ of $B$. 
If $a \wedge b = 0$ then 
the set of solutions can also be described using the 
{\em reproductive solution}, $F : B \rightarrow B$ defined by 
$F(x) = (b \vee x) \wedge \neg a$, which has the 
properties: 
\begin{itemize}
\item[(a)] $F(x) \in S$ for every $x \in B$ 
(i.e.\ $F$ is a parametric solution);
\item[(b)] For every $x \in S$ we have $F(x) = x$.
\end{itemize} 
For problems 
of the form $\exists x_1 \dots, \exists x_n (f(x_1. \dots, x_n)
\approx 0)$, where $f$ is an $n$-ary Boolean function,
the variables can be eliminated successively and
the reproductive solutions can be reconstructed. 
If $S$ is the set of all solutions in an algebra $B$, 
a reproductive solution is a map $F : B^n \rightarrow B^n$ with
\begin{itemize}
\item[(a)]  $F(x) \in S$ for every $x \in B^n$; and 
\item[(b)] For every $x \in S$ we have $F(x) = x$.
\end{itemize} 
These ideas have been generalized in \cite{Nipkow} to varieties generated
by primal algebras and in \cite{Burris92} to discriminator 
varieties, i.e.\ classes of algebras generated by a subset $K$ of algebras 
allowing a {\em switching term}, i.e.\ a term $s(x, y, u, v)$ 
with the property that for every algebra $\A \in K$, and all elements
$a, b, c, d \in \A$ we have $s(a, b, c, d) = c$ if $a = b$ and 
$s(a, b, c, d) = d$ if $a \neq b$. 
Nipkow \cite{Nipkow} investigated the links between reproductive
solutions and most general unifiers for varieties generated
by primal algebras, and 
Burris \cite{Burris92} proved that discriminator varieties have unitary unification
and the most general unifiers can be expressed 
using the discriminator term (cf.\ Theorem~\ref{burris} on page~\pageref{burris}). 
However, these results consider only unification problems 
without parameters. In this paper we investigate the existence of most
general solutions and reproductive solutions for more general
parametric problems of the form~(\ref{eq1}). 

\smallskip
While there are many examples of discriminator varieties 
relevant for computer science (cf.\ Example~\ref{ex-discr-var}), 
the signature of many theories interesting in computer science is not 
purely equational (i.e.\ contains also other
predicate symbols in addition to equality), and discriminator terms do
not necessarily exist. 

 In this paper we identify examples of logical 
theories allowing ``{\sf if}-{\sf then}-{\sf else}'' constructions 
on some of their models,
which generalize the discriminator terms mentioned above.
Having such operations is desirable in many cases: 
In the presence of parameters, solutions of equalities or inequalities in linear arithmetic 
are often described using case distinctions; also in applications 
we often consider functions defined using case distinctions.
We identify situations in which (parametric)  
solutions of conjunctions of literals exist, and use these results for
obtaining ways of describing most general solutions of such
constraints, and to analyze possibilities of deciding if solutions 
depending only on a set of parameters exist. 
\noindent The main results can be summarized as follows: 
\begin{itemize}
\item We extend the results on unification in discriminator
  varieties in \cite{Burris92} to a form of
  unification with constants (or with parameters).
\item We prove that, in extensions of 
 theories allowing quantifier elimination with ``{\sf if}-{\sf
   then}-{\sf else}''  constructions, most general solutions of existential formulae
can be described by terms. 
\item We give a criterion for deciding if for existential
  formulae with free variables ${\overline y}$ 
w.r.t.\ theories allowing quantifier elimination there are solutions which do
  not depend on a subset of the free variables in ${\overline y}$. 
\item We extend these results to a class of second-order
  quantifier elimination problems: We show how to find most general solutions 
  of formulae with existentially quantified function symbols in such cases.
\end{itemize}
These results contribute to better understanding the links  between satisfiability 
checking for parametric problems, symbol elimination, synthesis 
and unification with constants. This is the extended version of \cite{sofronie-ijcar2026}.

\smallskip
\noindent {\em The paper is structured as follows:} In Section~\ref{prelim} we
introduce the main notions needed in the paper (logical theories,
algebra, unification, and in particular unification in discriminator
varieties), and prove that the results on unification in discriminator
  varieties in \cite{Burris92} can be extended to a form of
  unification with constants. In Section~\ref{sect:ite} 
we discuss possibilities of introducing ``{\sf if}-{\sf then}-{\sf else}''
operators. In Section~\ref{sect:solutions-mgs}
we show how 
most general solutions can be constructed for a class of 
parametric existential problems in classes of theories allowing 
quantifier elimination in the presence of ``{\sf if}-{\sf then}-{\sf
  else}''-operators. In Section~\ref{sect:some-vars} give a criterion for deciding 
if for existential formulae with free variables ${\overline y}$ 
there are solutions which do 
  not depend on some of the variables in ${\overline y}$.  In Section~\ref{sect:soqe}
we show how these ideas can be extended to the analysis of 
certain second-order problems 
of the form $\exists f \,G(x_1, \dots, x_n)$, where $G$ is a
conjunction of flat clauses.  In Section~\ref{conclusion} we present 
the conclusions and plans for future work.

\

\section*{Table of contents}

\contentsline {section}{\numberline {1}Introduction}{1}{section.1.1}%
\contentsline {section}{\numberline {2}Preliminaries}{4}{section.1.2}%
\contentsline {subsection}{\numberline {2.1}Logic and logical theories}{4}{subsection.1.2.1}%
\contentsline {subsection}{\numberline {2.2}Algebraic structures}{4}{subsection.1.2.2}%
\contentsline {subsection}{\numberline {2.3}Unification}{5}{subsection.1.2.3}%
\contentsline {subsection}{\numberline {2.4}Unification in discriminator varieties}{6}{subsection.1.2.4}%
\contentsline {section}{\numberline {3}If-then-else operations}{9}{section.1.3}%
\contentsline {section}{\numberline {4}Solutions and most general solutions}{11}{section.1.4}%
\contentsline {subsection}{\numberline {4.1}Solution of a parametric existential formula}{12}{subsection.1.4.1}%
\contentsline {subsection}{\numberline {4.2}Most general solutions}{13}{subsection.1.4.2}%
\contentsline {section}{\numberline {5}Solutions depending on given variables}{18}{section.1.5}%
\contentsline {section}{\numberline {6}Second-order quantifier elimination}{20}{section.1.6}%
\contentsline {section}{\numberline {7}Conclusion}{25}{section.1.7}%

\section{Preliminaries}
\label{prelim}

We here introduce the main notions 
needed in the paper. 

\medskip
\noindent {\em Notation:} In what follows we will denote with (indexed
versions of) 
${\overline x}$ sequences of variables $x_1, \dots,
x_n$, and will sometimes write $\exists {\overline x} \phi$ for $\exists
x_1, \dots, x_n \phi$.  


\subsection{Logic and logical theories}
We assume known standard definitions from first-order logic  
such as structures, models, 
entailment,  satisfiability, unsatisfiability.  

We consider signatures of the form $\Pi = (\Sigma, {\sf Pred})$, 
where $\Sigma$ is a family of function symbols and ${\sf Pred}$
a family of predicate symbols. 
Let $X$ be a set of variables. 
We assume known notions such as terms and formulae over $\Pi$ 
with variables $X$. 
A $\Pi$-structure is a tuple 
$\A = (A, \{ f_{\A} \}_{f \in \Sigma}, \{p_{\A} \}_{p \in {\sf Pred}})$, 
where $A$ is a nonempty set (the
  universe of $\A$, denoted also $|\A|$) and for every $f \in \Sigma$ with
  arity $n$, $f_{\A} : A^n \rightarrow A$, and for every $p \in {\sf
    Pred}$ with arity $m$, $p_{\A} \subseteq A^m$. 
We assume known notions such as evaluation of a $\Sigma$-term $t$ in a structure
$\A$ w.r.t.\ a valuation $\beta : X \rightarrow \A$ (denoted
$\A(\beta)(t)$) and the evaluation of a $\Pi$-formula $\phi$ in $\A$
w.r.t.\ $\beta$ (notation: $\A(\beta)(\phi)$). We say that $\phi$ is
true (or holds) in $\A$ w.r.t.\ $\beta$ (notation: $(\A, \beta)
\models \phi$) if $\A(\beta)(\phi)$ is true. 
We denote
``falsum'' with $\perp$, so a formula $\phi$ is
unsatisfiable iff $\phi \models \perp$. 

In this paper, a theory $\T$ is described by a set of closed formulae 
(the axioms of the theory).  If $F$ and $G$ are formulae we 
write $F \models G$ (resp. $F \models_{\cal T} G$) 
to express the fact that every model of $F$ 
(resp. every model of $F$ which is also a model of 
$\T$) is a model of $G$. 
The definitions can be extended in a natural way to the case when 
$F$ is a set of formulae. Then $F \models_{\cal T} G$ if and only if 
${\cal T} \cup F \models G$. 
$F \models \perp$ means that $F$ is
unsatisfiable; $F \models_{\T} \perp$ means that there is no model of
$\T$ which is also a model of $F$. If there is a model of $\T$ which is also a
model of $F$ we say
that $F$ is $\T$-consistent.
We say that $F$ and $G$ are equivalent (notation: $F \equiv G$)
 iff $F \models G$ and $G
\models F$; $F$ and $G$ are equivalent 
w.r.t.\ $\T$ (notation: $F \equiv_{\T} G$) 
if $F \models_{\T} G$ and $G \models_{\T} F$.

A theory $\T$ with signature  
$\Pi$ {\em allows quantifier elimination} if for every $\Pi$-formula
$\phi$ there exists a quantifier-free $\Pi$-formula $\phi^*$
equivalent to $\phi$ w.r.t.\ $\T$.

\subsection{Algebraic structures} 
In the case in which the only
predicate symbol is $\approx$ (equality), we can assume that 
$\Sigma$ consists only of function symbols; in this case 
$\Sigma$-structures are called 
$\Sigma$-algebras. 
If $\A = (U_{\A}, \{ f_{\A} \}_{f \in \Sigma})$ is a 
$\Sigma$-algebra, and $t(x_1, \dots, x_n)$ is a term 
containing variables $x_1, \dots, x_n$, we will denote by 
$t_{\A} : U_{\A}^n \rightarrow U_{\A}$ the associated function 
defined for all tuples $(a_1, \dots, a_n)$ by 
$t_{\A}(a_1, \dots, a_n) = \A(\beta)(t)$, where $\beta : X \rightarrow
U_{\A}$ is the valuation with $\beta(x_i) = a_i$ for every $i \in \{
1, \dots, n \}$.

A class $K$ of algebras satisfies an equation or implication 
$\gamma$  (notation:
$K \models \gamma$)
if for every algebra $\A \in K$ we have $\A \models \gamma$. 
The equational theory of a class of structures is the set of universal 
atomic formulas that hold in all members of the class. 
For a class of algebras, this is simply the collection of all equations 
that hold in all members of the class.
An equational class (or variety) 
is the class of all algebras in which a given set $E$ of 
equations holds. 
We say that a variety ${\cal V}$ is generated by a subset $K \subseteq
{\cal V}$ (notation: ${\cal V} = {\cal V}(K) = HSP(K)$) 
 iff  every algebra in ${\cal V}$ is \underline{h}omomorphic to a \underline{s}ubalgebra of a
  \underline{p}roduct of algebras in $K$. For further definitions cf.\ \cite{BurrisSankappanavar1981}.

\subsection{Unification} 
We present some notions and results on unification. 
The problems discussed in this paper are related to 
unification with constants, so 
we restrict to this subclass of unification problems. 
For more  general definitions cf.\ 
\cite{BaaderSnyderHandbook99,BaaderSchulz96}. 

Let ${\cal V}$ be an equational theory axiomatized by a set $E$ of
equations, $\Sigma$  
its signature, and $\Delta$ a signature containing $\Sigma$ such that 
$\Delta \backslash \Sigma$ is a set of new constant symbols. 
Let ${\cal S} : \{ s_1 = t_1, \dots,  s_k = t_k \}$ be a system of equations, 
where $s_i, t_i \in T_{\Delta}(Y)$, the set of terms over the
signature $\Delta$ with variables in $Y$. 
Then ${\cal S}$ defines an {\em $E$-unification problem with free
  constants} in $\Delta$. 
Unification with free constants is related to 
unification  in which the terms contain free variables 
\cite{BaaderSnyderHandbook99,BaaderSchulz96,Baader98}.

A unification problem ${\cal S}$ with free constants in $\Delta$ 
{\em has a solution} with respect to   
${\cal V}$ if there is a substitution 
$\sigma : Y \rightarrow T_{\Delta}(Y)$ such that 
${\cal V} \models \sigma(s_i) \approx \sigma(t_i)$ for every $1 \leq i \leq k$. 
If this is the case, $\sigma$ is a solution (or a
unifier) for ${\cal S}$.\footnote{Note that if $\Sigma$ is a
  signature, $C$ a set of additional constants, and $Y$ a set of
  variables, $T_{\Sigma \cup C}(Y) = T_{\Sigma}(Y \cup C)$.}

${\cal S}$ is an $E$-unification problem with 
{\em linear constant restrictions}
if and only if it is an $E$-unification problem with constants and, in addition, 
a linear ordering $<$ on the variables and free constants occurring in 
${\cal S}$ is given. If ${\cal S}$ is an $E$-unification problem with linear constant restrictions, 
a solution for ${\cal S}$ is a substitution 
$\sigma : Y \rightarrow T_{\Delta}(Y)$
with the additional property that for every variable $y \in Y$ and 
every constant $c \in \Delta \backslash \Sigma$, if $y < c$ then $c$ does not occur in $\sigma(y)$.
The importance of $E$-unification with linear constant restrictions
is justified by its link with the {\em positive theory of $E$}, 
i.e.\ the 
collection of those closed formulae valid in the class of all 
models of $E$ which are (equivalent to a formula) of the form 
\begin{eqnarray}
\phi = Q_1 x_1 \dots Q_k x_k (\bigvee_{i=1}^n (s_{i1} = t_{i1} \wedge \dots
\wedge s_{i{m_i}} = t_{i{m_i}})),
\label{positive}
\end{eqnarray}
where $Q_1, \dots, Q_k \in \{ \exists, \forall \}$, and $x_1, \dots, x_k$
are distinct variables.

\begin{theorem}[\cite{BaaderSchulz96}]
If $E$-unification with linear constant restrictions 
is decidable then the positive theory of $E$ is
decidable.  
\label{proposition-baader-positive} 
\end{theorem}

A substitution $\sigma$ is more general w.r.t.\ ${\cal V}$ than 
a substitution $\sigma'$ if there is a substitution $\delta$ such that 
${\cal V} \models \sigma'(x) \approx \delta(\sigma(x))$ for every
variable $x \in X$. A substitution $\sigma$ is a most general 
solution (or a most general unifier) for ${\cal S}$  w.r.t.\ ${\cal V}$ if 
it is a solution for ${\cal S}$  w.r.t.\ ${\cal V}$ and is more general
than any other solution for ${\cal S}$ w.r.t.\ ${\cal V}$.

 One can study decidability of unifiability, 
the existence of unifiers, and their classification
according to generality.
If every unification problem has a most general unifier w.r.t.\ ${\cal
  V}$, we say that ${\cal V}$ has unitary unification.

\subsection{Unification in discriminator varieties} 
A  term $t(x, y, z)$ is a discriminator term for an algebra $\A$ 
if for all $a, b, c \in |\A|$,
$$t(a, b,c) = \left\{ \begin{array}{ll} c & \text{ if } a = b \\
 a & \text{ if } a \neq b.
\end{array} \right.$$
A term $s(x, y, u, v)$ is a switching term for an algebra $\A$ if 
for all $a, b, c, d \in |\A|$
$$s(a,b,c,d)= \left\{ \begin{array}{ll} c &\text{ if } a = b \\
d & \text{ if } a \neq b. 
\end{array} \right.$$
From a discriminator term we can construct a switching term 
and vice-versa, by

\smallskip
$
s(x, y, u, v)  =  t(t(x, y, u),t(x, y, v),v) \quad  \text{ and } \quad 
t(x, y, z)  =  s(x, y, z, x).
$

\smallskip
\noindent A variety ${\cal V}$ of algebras is a discriminator variety 
if there is a class $K$ of algebras
which generates ${\cal V}$ such that there is a ternary term $t(x, y, z)$ 
which is a discriminator term for every member of $K$. 
If ${\cal V}$ is a discriminator variety 
then $t(x, y, z)$ is a discriminator term for
precisely the simple algebras\footnote{A simple algebra is an algebra
  which has exactly two congruences, the total congruence and the
  identity.} in ${\cal V}$ (cf.\ \cite{Burris92}).
\begin{ex}
We give some examples of discriminator varieties: 
\begin{enumerate}
\item Boolean algebras ${\cal V}(B_2)$, where $B_2 = (\{ 0, 1 \}, \wedge,
  \vee, ', 0, 1)$ is the 2-element Boolean algebra.  

$t(x, y, z) = (x \wedge z) \vee (x \wedge y'
  \wedge z') \vee (x' \wedge y' \wedge z)$ is a discriminator term in
  the only simple Boolean algebras - the 1- and 2-element Boolean
  algebras. 
\item Boolean rings ${\cal V}({\mathbb Z}_2)$, where ${\mathbb Z}_2$ is the
  2-element ring. 

$t(x, y, z) = (x + y) \cdot x + (1 + x + y) \cdot z$ is a discriminator term in
  the only simple Boolean rings - the 1- and 2-element Boolean
  rings.
 
\item ${\cal V}({\mathbb Z}_p)$, for $p$ a prime number.

$t(x, y, z) = (x - y)^{p-1} \cdot x + (1 - (x - y)^{p-1}) \cdot z$ is a discriminator term in
  the only simple $p$-rings - the trivial ring and ${\mathbb Z}_p$.
\end{enumerate}
For further examples we refer to \cite{Burris92}.
\label{ex-discr-var}
\end{ex}

\begin{thm}[\cite{Burris92}]
Let ${\cal V}$ be a discriminator variety with signature $\Sigma$,
with switching term $s(x, y, u, v)$
on the simple algebras. 
Consider the unification problem
$\exists x_1 \dots x_n (p(x_1, \dots, x_n) \approx q(x_1,
\dots, x_n))$, where $x_1, \dots, x_n$ are variables in a set $X$.
Assume that the unification problem has a solution, and let
$t_1,\dots, t_n$ be terms in $T_{\Sigma}(X)$ such that 
${\cal V} \models p(t_1, \dots, t_n) = q(t_1, \dots, t_n)$. 
Then the substitution $\sigma : \{ x_1, \dots, x_n \} \rightarrow T_{\Sigma}(X)$ 
defined by: 
$$ \sigma(x_i) = s(p, q, x_i, t_i) \text{ for } 1 \leq i \leq n,$$
is a most general unifier.
\label{burris}
\end{thm}
The proof given in \cite{Burris92} relies on the following facts: 
\begin{itemize}
\item Any discriminator variety ${\cal V}$ is generated by the class
${\cal V}_S$
of its simple algebras. 
\item For any discriminator variety ${\cal V}$, every algebra ${\cal A}
\in {\cal V}$ 
and for every two terms $a, b \in |\A|$  we have 
$a = b$ iff ($(a, b) \in \theta$ for every maximal congruence
$\theta$ on $\A$).
\end{itemize} 
Using the fact that in the presence of a discriminator 
term every conjunction of atoms can 
be equivalently rewritten (w.r.t.\ the class of nontrivial 
simple algebras of ${\cal V}$) as an equality between terms, it is 
concluded that discriminator varieties have unitary unification \cite{Burris92}.

\medskip 
\noindent 
We show that the proof of Theorem~\ref{burris} given in
\cite{Burris92} can be adapted to unification 
problems of the form 
\begin{equation}
\exists x_1 \dots \exists x_n \, (p(x_1, \dots,x_n, {\overline y}) \approx q(x_1, \dots, x_n, {\overline
  y})) \label{eq3}
\end{equation}
in which the variables ${\overline y} = y_1, \dots, y_m$ are regarded as ``parameters''
(problems of type~(\ref{eq3}) can be regarded as unification
with constants).
\begin{thm}
Let ${\cal V}$ be a discriminator variety with switching term $s(x, y,
u, v)$ on its simple algebras. Consider the unification
problem~(\ref{eq3}). Assume that it has a solution and let
$t_1, \dots, t_n$ be terms 
such that 
$$ {\cal V} \models \forall {\overline y} \, (p(t_1, \dots,
t_n, {\overline  y}) \approx q(t_1, \dots, t_n,{\overline  y})).$$ 
Then the following hold (where ${\overline x}$ is $x_1, \dots, x_n$
and ${\overline y}$ is $y_1, \dots, y_m$):  
\begin{itemize}
\item[(i)] The substitution defined by 
$\sigma(x_i) = s(p({\overline x}, {\overline y}), q({\overline x},
{\overline y}), x_i, t_i)$ for $1 \leq i \leq n$, and 
$\sigma(y) = y$ for every other variable, 
is a most general unifier.  

\item[(ii)] For every algebra $\A \in {\cal V}$ and every $\beta : \{
  y_1, \dots, y_m \} \rightarrow \A$ let 

\medskip
\quad $\begin{array}{rcl} 
S_{\A, \beta} & = &  \{ {\overline
    a} \in |\A|^n \mid \A(\beta^{\overline a})(p) =
\A(\beta^{\overline a})(q) \} \\
& = & \{ {\overline
    a} \in |\A|^n \mid p_{\A}({\overline a}, \beta(y_1), \dots,
  \beta(y_m)) = q_{\A}({\overline a}, \beta(y_1), \dots,
  \beta(y_m)) \},
\end{array}$ 

\medskip
\noindent and let $F : |\A|^n \rightarrow |\A|^n$ be defined
  for all ${\overline a} = (a_1, \dots, a_n) \in |\A|^n$ by: 

\medskip
\quad $ F(a_1, \dots, a_n) = (\A({\beta}^{\overline a})(\sigma(x_1)), \dots,
\A({\beta}^{\overline a})(\sigma(x_n))),$

\medskip
\noindent where $\beta^{\overline a}(y_i) = \beta(y_i)$ for all $1 \leq i \leq
m$ and $\beta^{\overline a}(x_i) = a_i$ for all $1 \leq i \leq n$. 

Then the following hold: 
\begin{itemize}
\item[(a)] $F(|\A|^n) \subseteq  S_{\A, \beta}$, and 
\item[(b)] for every
$(a_1, \dots, a_n) \in S_{\A, \beta}$ we have $F(a_1, \dots, a_n) =
(a_1, \dots, a_n)$. 
\end{itemize}
Hence, $F(|\A|^n) = S_{\A, \beta}$. 
\end{itemize}
\label{discr-constants}
\end{thm}
{\em Proof:} 
(i) Let $\A \in {\cal V}$ and $\beta : X \rightarrow \A$. 
Let $\rho$ be a maximal congruence on $\A$. Then in 
the simple algebra $\A/\rho$ the term $s(x, y, u, v)$ is a
switching term, so we have (we denote by $a/\rho$ the equivalence
class of $a$ in $\A/\rho$): 

\smallskip
$\A(\beta)(\sigma(x_i))/{\rho} = \left\{ \begin{array}{ll}
\A(\beta)(x_i)/{\rho} & \text{ if } \A(\beta)(p({\overline
  x}, {\overline y}))/\rho = \A(\beta)(q({\overline x}, {\overline y}))/\rho \\
\A(\beta)(t_i)/{\rho} & \text{ otherwise }
\end{array} \right.$

\smallskip
\noindent Then, for every  maximal congruence $\rho$ on $\A$: 

\smallskip
$\A(\beta)(\sigma(p))/{\rho}  =  \A(\beta)(p(\sigma(x_1), \dots, \sigma(x_n),
{\overline y}))/{\rho}$ 

\nopagebreak
$\begin{array}{lcl}
~~~~~~~& = & p_{\A/\rho}(\A(\beta)(\sigma(x_1))/{\rho}, \dots, \A(\beta)(\sigma(x_n))/{\rho},
\A(\beta)(y_1)/{\rho}, \A(\beta)(y_m)/{\rho}) \\[0.5ex]
& = & \left\{ \begin{array}{ll}
\A(\beta)(p(x_1, \dots, x_n, {\overline
  y}))/\rho  & \text{ if } \A(\beta)(p({\overline x}, {\overline y}))/{\rho} =
\A(\beta)(q({\overline x}, {\overline y}))/{\rho} \\
\A(\beta)(p(t_1, \dots, t_n, {\overline y}))/\rho & \text{ otherwise }
\end{array} \right. \\[3ex]
& = &  \left\{ \begin{array}{ll}
\A(\beta)(q(x_1, \dots, x_n, {\overline
  y}))/\rho  & \text{ if } \A(\beta)(p({\overline x}, {\overline y}))/{\rho} =
\A(\beta)(q({\overline x}, {\overline y}))/{\rho} \\
\A(\beta)(q(t_1, \dots, t_n, {\overline y}))/\rho & \text{ otherwise }
\end{array} \right. \\[3ex]
& = & \A(\beta)(\sigma(q))/{\rho}. 
\end{array}$

\smallskip
\noindent So $\A(\beta)(\sigma(p))/\rho = \A(\beta)(\sigma(q))/\rho$ for every 
maximal congruence $\rho$ on $\A$. 

\noindent It follows that for every $\A \in {\cal V}$ and every $\beta: X
\rightarrow \A$ we have $\A(\beta)(\sigma(p)) = \A(\beta)(\sigma(q))$, 
so ${\cal V} \models \sigma(p) = \sigma(q)$. 

\smallskip
\noindent 
Let $\mu : X \rightarrow T_{\Sigma}(X)$ be another solution of
problem $(\ref{eq3})$. 
Then ${\cal V} \models \mu(p) \approx \mu(q)$, 
i.e.\ for every algebra $\A \in {\cal V}$ 
and every $\beta : X \rightarrow \A$ we have 
$\A(\beta)(\mu(p)) = \A(\beta)(\mu(q))$. 

\noindent Then $\mu(\sigma(y)) = \mu(y)$ for all $y \in
X \backslash \{ x_1, \dots, x_n \}$, and for all $1 \leq i \leq n$ we have 

\smallskip
$\mu(\sigma(x_i)) = \mu(s(p, q, x_i, t_i))  = s(\mu(p), \mu(q),
\mu(x_i), \mu(t_i)) $, 

\smallskip
\noindent so, for every $\A \in {\cal V}$, for every valuation $\beta: X
\rightarrow \A$, and every maximal congruence $\rho$ on $\A$ we have: 

\smallskip
$\A(\beta)(\mu(\sigma(x_i)))/\rho = \A(\beta)(s(\mu(p), \mu(q),
\mu(x_i), \mu(t_i)))/\rho = \A(\beta)(\mu(x_i))/\rho$.

\smallskip
\noindent Therefore, ${\cal V} \models \mu(\sigma(x_i)) = \mu(x_i)$. 
Since for all the other variables $y$ we have $\sigma(y) = y$, 
${\cal V} \models \forall x (\mu(\sigma(x)) = \mu(x))$. 

\medskip
\noindent (ii) (a) Let $\A$ and $\Sigma$ be as in the statement of
(ii). Let $(a_1, \dots, a_n) \in |\A|^n$. 
We know that ${\cal V} \models \sigma(p) \approx \sigma(q)$, i.e.\
$${\cal V} \models p(\sigma(x_1), \dots, \sigma(x_n), y_1, \dots, y_m)
\approx  q(\sigma(x_1), \dots, \sigma(x_n), y_1, \dots, y_m).$$ 
Therefore, $\A(\beta^{\overline a})(p(\sigma(x_1), \dots,
\sigma(x_n), {\overline y}))
=  \A(\beta^{\overline a})(q(\sigma(x_1), \dots, \sigma(x_n),
{\overline y}))$. This means that  
$F(a_1, \dots, a_n) = (\A(\beta^{\overline a})(\sigma(x_1)), \dots,
\A(\beta^{\overline a})(\sigma(x_n))) \in S_{\A, \beta}$.

\smallskip
\noindent Conversely, let  $(a_1, \dots, a_n) \in
S_{\A, \beta}$. 
Then $\A(\beta^{\overline a})(p) = \A(\beta^{\overline  a})(q)$, i.e.:  

\smallskip
$p_{\A}(a_1, \dots, a_n, \beta(y_1), \dots, \beta(y_n)) = q_{\A}(a_1, \dots,
a_n, \beta(y_1), \dots, \beta(y_n)).$

\smallskip
\noindent We know that $F(a_1, \dots, a_n)  =  (\A({\beta}^{\overline a})(\sigma(x_1)), \dots,
\A({\beta}^{\overline a})(\sigma(x_n)))$.
We show that for every maximal congruence 
$\rho$ on $\A$,  
we have 

$(\A({\beta}^{\overline a})(\sigma(x_1))/\rho, \dots,
\A({\beta}^{\overline a})(\sigma(x_n))/\rho) = (a_1/\rho, \dots, a_n/\rho)$.

\medskip
\noindent Indeed, $(\A({\beta}^{\overline a})(\sigma(x_1))/\rho, \dots,
\A({\beta}^{\overline a})(\sigma(x_n))/\rho) = $ 

$\begin{array}{rcl}
& =  & \left\{ \begin{array}{ll}
(\A(\beta)(x_1)/\rho, \dots, \A(\beta)(x_n)/\rho) & \text{ if }
  \A(\beta^{\overline a})(p) /\rho  = \A(\beta^{\overline a})(p) /\rho \\
(\A(\beta)(t_1)/\rho, \dots, \A(\beta)(t_n))/\rho & \text{ otherwise }
\end{array} \right. \\
& =  & (a_1/\rho, \dots, a_n/\rho)
\end{array}$

\smallskip
\noindent It follows therefore that $F(a_1, \dots, a_n) = (a_1, \dots, a_n)$.
\QED
\begin{defi}
Let ${\cal V}$ be a variety. Consider the unification
problem~(\ref{eq3}). Let $X$ be a set of variables containing all 
variables occurring in the unification problem (\ref{eq3}).
Let $\A \in {\cal V}$, and $\beta : X \rightarrow \A$, 
and let $S_{\A, \beta}$ be the set of solutions of (\ref{eq3})
as defined in Theorem~\ref{discr-constants}(ii). 
We say that a function $F : |\A|^n \rightarrow |\A|^n$ is a {\em
  reproductive solution} in $\A$ w.r.t.\ $\beta$ if it has the 
property 
that 
\begin{itemize}
\item[(a)] $F(|\A|^n) \subseteq S_{\A, \beta}$, and 
\item[(b)] for every 
${\overline a} \in S_{\A, \beta}$ we have $F({\overline a}) =
{\overline a}$. 
\end{itemize}
\end{defi}
This notion of reproductive solution generalizes the definitions of 
reproductive solution for Boolean algebras \cite{Rudeanu} 
and for primal algebras \cite{Nipkow}. Theorem~\ref{discr-constants} shows 
how most general unifiers and reproductive solutions can be
constructed for discriminator varieties.

\section{If-then-else operations}
\label{sect:ite}
We would like to prove results similar to those established for discriminator
varieties for other theories. We analyze problems of the form 

\medskip
\quad \quad \quad \quad $\exists x_1 \dots x_n \phi(x_1, \dots, x_n, y_1, \dots, y_m)$, \hfill (\ref{eq1})

\medskip
\noindent where $\phi$ is a quantifier-free formula with free variables
$x_1, \dots, x_n, y_1, \dots, y_m$ and  analyze the existence of most
general solutions w.r.t.\ a theory $\T$. 

\noindent In general, discriminator terms do not 
exist, even for selected models of the theory $\T$. 
We will use 
 ``{\sf if}-{\sf then}-{\sf else}''-operations which behave like the
 switching terms; usually such operations are easier to evaluate in certain
models.  
\begin{ex}
Let ${\mathbb R}[X_1, \dots, X_n]$ be the ring of polynomials in 
indeterminates $X_1, \dots, X_n$. Clearly, two polynomials $p$ and $q$
are equal iff they have the same coefficients iff the associated
functions are equal. Therefore, a ``switching operation'' 
which tests the equality between polynomials (as formal
objects) would be difficult to evaluate, because it would require a
form of universal quantification. 
 
However, for all $a_1, \dots, a_n \in {\mathbb R}$ 
we can consider the maximal congruence on 
${\mathbb R}[X_1, \dots, X_n]$ described by the ideal 
$(X_1 - a_1, \dots, X_n - a_n)$. 
The quotient ring 
${\mathbb R}[X_1, \dots, X_n]/(X_1 - a_1, \dots, X_n - a_n)$ is simple 
(i.e.\ a simple algebra). 

Since the ideal $(X_1 - a_1, \dots, X_n - a_n)$ is the kernel of
the evaluation map ${\sf ev}$ which maps every polynomial $p$ into 
the value $p(a_1, \dots, a_n)$ and the equivalence
class $[p]$ of a polynomial $p$ in ${\mathbb R}[X_1, \dots,
X_n]/(X_1 - a_1, \dots, X_n - a_n)$ can be identified with the value
$p(a_1, \dots, a_n)$, there is a ring isomorphism: 
$${\mathbb R}[X_1, \dots, X_n]/(X_1 - a_1, \dots, X_n - a_n) \simeq
{\mathbb R}.$$ 
\noindent The ${\sf ite}_{\approx}$ operation on ${\mathbb R}[X_1, \dots,
X_n]/(X_1 - a_1, \dots, X_n - a_n)$ can be defined by: 

\smallskip
${\sf ite}_{\approx}(p(a_1, \dots, a_n), q(a_1, \dots, a_n), c, d) = \left\{ \begin{array}{ll}
c & \text{ if } p(a_1, \dots, a_n) = q(a_1, \dots, a_n) \\
d & \text{ otherwise} 
\end{array} \right.$ 

\smallskip
\noindent and can be regarded as an operation in (the simple ring)
${\mathbb R}$.
\hfill $\blacksquare$
\label{ex:polynomials}
\end{ex}
\begin{rem}
We can define operations ${\sf ite}_{\approx}$ and ${\sf ite}_{\leq}$
also on the ring of polynomials, by 
$${\sf ite}_{\approx}(p, q, r, s) = \left\{ \begin{array}{ll}
r & \text{ if } p = q \\
s & \text{ if } p \neq q \end{array}\right. \quad  {\sf ite}_{\approx}(p, q, r, s) = \left\{ \begin{array}{ll}
r & \text{ if } p \leq q \\
s & \text{ if } p \not\leq q \end{array}\right.$$
However, since $p = q$ iff the cofficients of $p$ and $q$ coincide 
iff $p(a_1, \dots, a_n) = q(a_1, \dots, a_n)$ for all $a_1,\dots,a_n
\in {\mathbb R}$; and 
$p \leq q$ iff  $p(a_1, \dots, a_n) \leq q(a_1, \dots, a_n)$ for all
$a_1,\dots,a_n \in {\mathbb R}$, these operations are in general 
more difficult to evaluate than the corresponding operations 
on ${\mathbb R}$ described in Example~\ref{ex:polynomials}.  
\end{rem}
Our goal is to obtain (possibly most general) solutions for 
problems of the form 
$$ \exists x_1, \dots \exists x_n \phi(x_1, \dots, x_n, {\overline
  y}),$$
i.e.\ find terms $t_1, \dots, t_n$ with the property that 
$$  \T \models  \phi(t_1, \dots, t_n, {\overline
  y}) \text{ or } \T \models \psi({\overline y}) \rightarrow \phi(t_1, \dots, t_n, {\overline
  y}),$$
where $\psi$ is a formula representing a condition under which solutions exist.
In many cases in the presence of parameters, we are not 
likely to find terms in the signature of the theory which are solutions for formulae of
form~(\ref{eq1}). 
\begin{ex} Let $\T = LI({\mathbb R})$, the theory of linear real
  arithmetic. Consider the problem $ \exists x (x \geq x_1 \wedge x
  \geq x_2)$. 
Clearly, $LI({\mathbb R}) \models \exists x (x \geq x_1 \wedge x
  \geq x_2)$, but one cannot find a term $t$
built in the signature of linear real arithmetic such that 

\medskip
~~~~~~~~~~~~~~~~~~~~~~~~$LI({\mathbb R}) \models t \geq x_1  \wedge t \geq x_2.$

\medskip
\noindent In the presence of {\sf if}-{\sf then}-{\sf else}-operations
on ${\mathbb R}$,
we can encode definitions by case distinctions as terms, and can
represent a solution as a term over an extended signature containing 
an {\sf if}-{\sf then}-{\sf else} construction: 

\smallskip
$ t = \ite{x_1 \geq x_2}{x_1}{x_2} = \left\{ \begin{array}{ll} x_1 & \text{ if } x_1 \geq x_2 \\
x_2 & \text{ otherwise.} \end{array} \right. $ 
\hfill $\blacksquare$
\end{ex}
Thus, if we can define ``{\sf if}-{\sf then}-{\sf else}'' constructions on suitable models 
of the theories, we are likely to describe the 
solutions and the most general solutions for problems of the form~(\ref{eq1}) using
terms over this extended signature.  

\bigskip
\noindent Let $\T$ be a theory with signature $\Pi = (\Sigma, {\sf
  Pred})$, and ${\cal F}$ a subset of the quantifier-free $\Pi$-formulae. 
In what follows we consider the following assumptions: 

\begin{description}
\item[(A1)] 
There exists a model $\A$ of $\T$ with 
the property that for every quantifier-free formula $\phi({\overline
  x}, {\overline y}) \in {\cal F}$ with free variables
${\overline x}, {\overline y}$ we have 
$\T \models \forall {\overline y} \exists {\overline x}\, \phi({\overline x}, {\overline y})$
iff $\A \models \forall  {\overline y} \exists {\overline x}\,
\phi({\overline x}, {\overline y})$.

\item[(A2)] $\T$ allows elimination of existential quantifiers for all
  formulae in  ${\cal F}$ (i.e.\ for every $\phi(x, {\overline y}) \in
  {\cal F}$ there exists a quantifier-free formula $\psi({\overline y})
  \in {\cal F}$ such
  that $\exists x \, \phi(x, {\overline y}) \equiv_{\T}
  \psi({\overline y})$).   
\end{description}
\begin{rem}
If (A1) holds, then for every $m$-ary 
predicate symbol  $p \in {\sf Pred}$ we can define 
an ${\sf ite}_p$ function on $\A$ for all $a_1, \dots,
a_m, a, b \in \A$ by:  
$${\sf ite}_p(a_1, \dots, a_m, a, b) = \left\{ \begin{array}{ll}
a & \text{ if } p_{\A}(a_1, \dots, a_m) \text{ is true }  \\
b &  \text{ if } p_{\A}(a_1, \dots, a_m) \text{ is false. } \end{array} \right.$$
The interpretation of ${\sf ite}_{\approx}(s_1, s_2, t_1, t_2)$
is the same as the semantics of a switching term. In (A1), $\A$ is usually chosen s.t.\
$p_{\A}(a_1, \dots, a_m)$ is easy to evaluate.
\end{rem}
\begin{ex}
We give some examples of theories $\T$ and formula classes ${\cal F}$
which satisfy the assumptions above: 
\begin{description}
\item[{\bf Bool}:] Let ${\sf Bool}$ be the theory of Boolean algebras, and ${\cal
    F}$ the class of all formulae in the signature of ${\sf Bool}$ of
  the form $\phi({\overline x}) = p({\overline x}) \approx
  q({\overline x}).$

The following hold:
\begin{itemize}
\smallskip
\item  ${\sf Bool} \models
\forall {\overline y} \exists {\overline x} \,(p({\overline x},
  {\overline y}) \approx q({\overline x}, {\overline y}))$ iff $B_2 \models
\forall {\overline y} \exists {\overline x} \, (p({\overline x},
  {\overline y}) \approx q({\overline x}, {\overline y}))$, where $B_2$ is the 2-element Boolean algebra (cf.\
  e.g.\ \cite{Baader98}). 
\item  The switching term satisfies the conditions of ${\sf
    ite}_{\approx}$. 
\item ${\sf Bool}$ allows quantifier elimination for formulae
in ${\cal F}$: Indeed every equation of form $p(x, {\overline y})
\approx q(x, {\overline y})$ can be written in the form 
$(a \wedge x) \vee (b \wedge \neg x) \vee c \approx 0$, where 
$a, b, c$ are terms containing the variables ${\overline y}$, and 
$\exists x ((a \wedge x) \vee (b \wedge \neg x) \vee c \approx 0)
\equiv_{\sf Bool} (a \wedge b) \vee c \approx 0 \in {\cal F}.$
\end{itemize}
\medskip
\item[{\bf LI}$({\mathbb R})$:] Let $\A = {\mathbb R}$, the model of
  $LI({\mathbb R})$
having as support the set of real numbers and the usual interpretation
for the function and predicate symbols and ${\cal F}$ be the class of 
all quantifier-free formulae in the signature of $LI({\mathbb R})$.
Then assumption (A1)
holds, because for every quantifier-free formula $\phi$:
$LI({\mathbb R}) \models \forall {\overline y} \exists {\overline x}
\phi({\overline x}, {\overline y})$ iff ${\mathbb R}
\models \forall {\overline y} \exists {\overline y} \phi({\overline x}, {\overline y})$. 
Since linear real arithmetic allows
quantifier elimination, assumption (A2) holds as well. 
The operations ${\sf ite}_{\approx}, {\sf ite}_{\leq}$ on ${\mathbb R}$
are defined, for every $a, b, c,
d \in {\mathbb R}$ by: 
$$\begin{array}{ll}
{\sf ite}_{\approx}(a, b, c, d) := \left\{ \begin{array}{ll}
c & \text{ if } a = b \\
d & \text{ if } a \neq b \end{array} \right. & \quad \quad {\sf ite}_{\leq}(a, b, c, d) := \left\{ \begin{array}{ll}
c & \text{ if } a \leq b \\
d & \text{ if } a > b. \end{array} \right.
\end{array}$$
\item[Theory of real closed fields:] We can choose
  $\A = {\mathbb R}$, the field of real numbers and ${\cal F}$ the set
  of all formulae in the signature of real closed fields. Then
  assumption (A1) holds, and since the theory of real closed fields 
  allows quantifier elimination assumption (A2) holds as well. \hfill $\blacksquare$
\end{description}
\label{ex:a1-a2}
\end{ex}
Let $\T$ be a theory satisfying assumption (A1). 
Then for every 
$m$-ary predicate symbol $p$ we can define 
a function ${\sf ite}_{\neg p}$ by 
${\sf ite}_{\neg p}({\overline x}, y, z) = {\sf
  ite}_p({\overline x}, z, y)$.  

For every conjunction of literals 
$\phi = p_1({\overline x}_1) \wedge \dots \wedge p_k({\overline
  x}_k)$, where $p_i \in \{ q_i, \neg q_i \}$, 
we can define a construction of the form 
$\ite{\phi}{y}{z}$ by:  

\smallskip
$\displaystyle{\ite{\bigwedge_{i = 1}^n p_i({\overline x}_i)}{y}{z} := 
{\sf ite}_{p_1}({\overline x}_1, {\sf ite}_{p_2}({\overline x}_2, \dots, 
{\sf ite}_{p_k}({\overline x}_k, y, z), \dots z), z).}$
\begin{defi}
Let $\T$ be a theory with signature $\Pi = (\Sigma, {\sf Pred})$ 
satisfying assumption (A1), let $\Sigma_{\sf ite} = \{ {\sf ite}_p
\mid p \in {\sf Pred} \}0$, let 
$\Sigma' = \Sigma \cup \Sigma_{\sf ite}$ and $\Pi' =
(\Sigma', {\sf Pred})$. 
Let $\A$ be the model in assumption (A1), let $\A'$ be the expansion 
of $\A$ with operations in $\Sigma_{\sf ite}$, and let $\T'$ be the
theory consisting of all
consequences of $\A'$ (a theory with signature $\Sigma'$). 
We refer to $\A'$ as the {\em {\sf ite}-expansion  of $\A$} and to 
$\T'$ 
as the {\em {\sf  ite}-expansion of $\T$}. 
\label{defi:tprim}
\end{defi}

\section{Solutions and most general solutions}
\label{sect:solutions-mgs}

We now define the notions of solution and of most general 
solution for a formula of the form: 

\medskip
$\exists x_1, \dots, x_n \phi(x_1, \dots, x_n, y_1, \dots, y_m), \hfill (\ref{eq1})$

\medskip
\noindent and identify situations in which the results established in
Theorem~\ref{discr-constants} for discriminator varieties 
can be generalized to more general theories.

Let $\T$ be a theory with signature $\Pi$ and let ${\cal F}$ be a
subset of the quantifier-free $\Pi$-formulae such that 
assumptions (A1) and (A2) hold. Let $\T', \Sigma', \Pi'$ and $\A'$ 
be as defined in Definition~\ref{defi:tprim}. We 
consider the following additional property: 
\begin{description}
\item[(A3)] For every $\phi(x_1, \dots, x_n, {\overline y}) \in {\cal
    F}$, if $\psi({\overline y}) \equiv_{\T} \exists x_1 \dots \exists x_n
  \phi(x_1, \dots, x_n, {\overline y})$ is a quantifier-free formula 
which exists by assumption (A2), 
 there exist $\Sigma'$-terms $t_1, \dots, t_n$ such that 
   $\T' \models (\psi({\overline y}) \rightarrow \phi(t_1, \dots, t_n,
   {\overline y}))$.
\end{description}
In what follows we restrict, for the sake of simplicity (w.l.o.g.), to problems
of form~(\ref{eq1}) for 
quantifier-free formulae $\phi({\overline x},
  {\overline y})$ which are
{\em conjunctions of literals}.  

\subsection{Solution of a parametric existential formula}

We propose the notion of solution and conditional solution for problems of the
form~(\ref{eq1}) w.r.t.\ a theory $\T$.

\begin{defi}[Solution]
A solution for a problem of the form~(\ref{eq1}) w.r.t.\ a theory $\T$ is a tuple of $\Pi$-terms 
$t_1, \dots, t_n$ such that ${\cal T} \models  \phi(t_1, \dots, t_n,
   {\overline y})$, or equivalently a substitution $\sigma : X
   \rightarrow T_{\Pi}(X)$ 
with $\sigma(x_i) = t_i$ for all $1 \leq i \leq n$ and $\sigma(y) =
y$ for all the other variables $y \in X$  such that ${\cal T} \models \sigma(\phi)$.

\smallskip
Assume that $\T$ satisfies condition (A1), and let 
$\Pi', \A'$ and $\T'$ be as defined in Definition~\ref{defi:tprim}. 
A solution with {\sf if}-{\sf then}-{\sf else}-terms for  a problem of
the form~(\ref{eq1}) w.r.t.\ a theory $\T$ is a tuple of $\Pi'$-terms 
$t_1, \dots, t_n$ such that ${\cal T}' \models  \phi(t_1, \dots, t_n,
   {\overline y})$, or equivalently a substitution $\sigma : X
   \rightarrow T_{\Pi'}(X)$ 
with $\sigma(x_i) = t_i$ for all $1 \leq i \leq n$ and $\sigma(y) =
y$ for all the other variables $y \in X$ 
such that ${\cal T}' \models \sigma(\phi)$.
\end{defi}
Solutions might only exist in some cases. If the theory $\T$
allows quantifier elimination, then there exists a quantifier-free 
formula $\psi$ containing only the variables ${\overline y}$ such that 
$$\psi({\overline y}) \equiv_{\T} \exists x_1 \dots \exists x_n
  \phi(x_1, \dots, x_n, {\overline y}).$$
A {\em conditional solution} for a problem of the form~(\ref{eq1}) 
w.r.t.\ $\T$ is a substitution with the property that 
for every model $\A$ of $\T$ and every assignment $\beta$
such that $(\A, \beta) \models \psi({\overline y})$ we have 
$(\A, \beta) \models \sigma(\phi)$. 
\begin{defi}[Conditional solution]
A {\em conditional solution} for a problem of the form~(\ref{eq1}) 
w.r.t.\ a theory $\T$ is a substitution $\sigma : X \rightarrow T_{\Pi}(X)$ 
with $\sigma(y_i) = y_i$ for all $1 \leq i \leq m$  such that 
${\cal T} \models (\psi({\overline y}) \rightarrow \sigma(\phi))$.

A {\em conditional solution} with {\sf if}-{\sf then}-{\sf else}-terms
for a problem of the form~(\ref{eq1}) 
w.r.t.\ a theory $\T$ is a substitution 
$\sigma : X \rightarrow T_{\Pi'}(X)$ 
with $\sigma(y_i) = y_i$ for all $1 \leq i \leq m$ such that 
${\cal T}' \models (\psi({\overline y}) \rightarrow \sigma(\phi))$.

\end{defi}

\begin{ex} We present some examples: 
\begin{description}
\item[{\bf Bool}:] The theory of Boolean algebras is a special case,
  since it has a switching term. We can therefore choose $\T' = {\sf
    Bool}$. Let $\phi(x, {\overline y}) = (a \wedge x) \vee (b \wedge \neg
  x) \approx 0$, where $a$ and $b$ are
  terms in the variables ${\overline y}$. We know that 
  $$\exists x ((a \wedge x) \vee (b \wedge \neg
  x) \approx 0) \equiv_{\sf Bool} a \wedge b \approx 0.$$ 
This means that $\phi$ might not have a solution in general.

However, it has a conditional solution: If $a \wedge b \approx 0$ then
e.g.\ the term $t = b$ is a solution (or, equivalently, the
substitution $\sigma$ with $\sigma(x) = b$, $\sigma(y_i) = y_i$ for
all the variables $y_i$ in ${\overline y}$ is a solution).

We have: For every Boolean algebra
  $B$ and every $\beta : X \rightarrow B$, if $(B, \beta) \models a
  \wedge b \approx 0$ then $(B, \beta) \models (a \wedge t) \vee (b
  \wedge \neg  t) \approx 0$, so 
$${\sf Bool} \models (a \wedge b
  \approx 0 \rightarrow \phi(t, {\overline y})), \quad \text{i.e.}
  \quad {\sf Bool} \models (a \wedge b \approx 0 \rightarrow \sigma(\phi)).$$
\item[{\bf LI}$({\mathbb R})$:] Linear real arithmetic allows quantifier
  elimination. For this we can, for instance, use methods such as
  virtual substitution \cite{WeispfenningLoos}, which allows to 
  compute a finite set $T$ of testpoints such that: 
$$\exists x \phi(x, {\overline y}) \equiv \bigvee_{a \in T} \phi(a,
{\overline y}).$$
$T$ contains all terms $u_i({\overline y})$ such that $\phi$ contains 
$x \approx u_i({\overline y})$ or $x \geq u_i({\overline y})$, 
and in addition (i) a term denoted
with $- \infty$, (ii) terms of the form $u_i({\overline y}) + \varepsilon$. 
\begin{itemize}
\item Instead of $- \infty$ we can choose the
$\Sigma'$-expressible term 
$t_m - 1$, where 
$t_m = {\sf min}(\{ s_i({\overline y}) \mid x \bowtie s_1({\overline y}) \text{
  occurs in } \phi \})$ (here $\bowtie \in \{ {\leq}, {<}, {\geq}, {>}, {\approx}, {\not\approx}\}$). 
\item We can replace $\varepsilon$ in the terms $u_i({\overline y}) + \varepsilon$ with $e/2$, where \\
$e = {\sf min}(|s_i - s_j| \mid s_i \neq s_j, x \,{\bowtie}\, s_i, x \,{\bowtie} \,s_j \,\text{occur in }
\phi, \bowtie\, \in \{{\leq}, {<}, {\geq}, {>}, {\approx}, {\not\approx} \})$. 
\end{itemize}
If $T = \{ t_1, \dots, t_k \}$, a solution can be represented by the
$\Pi'$-term:

\smallskip
$ 
{\sf if}~\phi(t_1, {\overline y})~{\sf then}~ t_1~{\sf else} 
 ~({\sf if}~\phi(t_2, {\overline y})~{\sf then}~ t_2~{\sf else}( 
 \dots 
 ({\sf if}~\phi(t_k, {\overline y})~{\sf then}~ t_k~{\sf else}~c_f) \dots))$ 

\smallskip
\noindent where $c_f$ is a special constant which stands for ``no
solution''. 

\medskip
\item[Theory of real closed fields:] Finding 
  testpoints $t_1, \dots, t_n$ which can be used for obtaining solutions 
  depends on the method for
  quantifier elimination used (e.g.\ 
cylindrical algebraic decomposition computes a decomposition of
${\mathbb R}^n$ into regions; the method needs to be adapted to
generate suitable testpoints). \hfill $\blacksquare$
\end{description}
\label{ex:a3}
\end{ex}

\subsection{Most general solutions}
We now analyze the existence of most general solutions
for problems of the form~(\ref{eq1}). 

\begin{defi}[Most general solution]
A substitution $\sigma : X \rightarrow T_{\Pi}(X)$ 
which is a solution 
for a problem  of the form~(\ref{eq1}) is a {\em most
  general solution} if it has the property that 
for every other solution $\mu$ there exists a substitution $\delta$ 
such that for every $x \in X$, ${\cal T} \models \mu(x) \approx
\delta(\sigma(x))$. 

A substitution $\sigma : X \rightarrow T_{\Pi'}(X)$ 
which is a solution with {\sf if}-{\sf then}-{\sf else}-terms 
for a problem  of the form~(\ref{eq1}) is a {\em most
  general solution} (with {\sf if}-{\sf then}-{\sf else}-terms) 
if it has the property that 
for every other solution $\mu$ with {\sf if}-{\sf then}-{\sf
  else}-terms  
there exists a substitution $\delta$ 
such that for every $x \in X$, ${\cal T}' \models \mu(x) \approx
\delta(\sigma(x))$. 
\end{defi} 
We can also define a notion of most general conditional solution: a conditional solution which is
more general than all other conditional solutions. Assume that the 
theory $\T$ allows quantifier elimination, and let $\psi$ be a
quantifier-free formula 
containing only the variables ${\overline y}$ with  
$\psi({\overline y}) \equiv_{\T} \exists x_1 \dots \exists x_n
  \phi(x_1, \dots, x_n, {\overline y}).$
\begin{defi}[Most general conditional solution]
A substitution $\sigma : X \rightarrow T_{\Pi}(X)$ which is a conditional solution 
for a problem  of the form~(\ref{eq1}) is a {\em conditional most
  general solution} if it has the property that 
for every other conditional solution $\mu$ 
there exists a substitution $\delta$ 
such that ${\cal T} \models \psi \rightarrow \forall x (\mu(x) \approx \delta(\sigma(x)))$. 

A substitution $\sigma : X \rightarrow T_{\Pi'}(X)$ 
which is a conditional solution with {\sf if}-{\sf then}-{\sf else}-terms
for a problem  of the form~(\ref{eq1}) is a {\em conditional most
  general solution} with {\sf if}-{\sf then}-{\sf else}-terms 
if it has the property that 
for every other conditional solution $\mu$ with 
{\sf if}-{\sf then}-{\sf else}-terms
there exists a substitution $\delta$ 
such that for every $x \in X$, ${\cal T}' \models \psi \rightarrow
\forall x (\mu(x) \approx \delta(\sigma(x)))$. 
\end{defi} 
In what follows, in order to keep the formulation simple, 
when we talk about {\em solutions} resp.\ {\em most general solutios} 
we will in general mean {\em conditional solutions with {\sf if}-{\sf
    then}-{\sf else}-terms} resp.\ 
{\em most general conditional solutions with {\sf if}-{\sf
    then}-{\sf else}-terms}.

\smallskip
We show that under assumptions (A1), (A2) and (A3) 
the most general solutions (with {\sf if}-{\sf then}-{\sf else}-terms)
for $\exists {\overline x} \phi({\overline x}, {\overline
    y})$ (if solutions exist), can be described as substitutions over the
  signature $\Pi' = \Pi \cup \Sigma_{\sf ite}$, which can also be used
  to describe  the reproductive solutions.

\begin{thm}
Let $\T$ be a theory satisfying assumptions (A1), (A2) and (A3),  
and $\T'$ be its {\sf ite}-expansion (cf. Definition~\ref{defi:tprim}).
Let $\phi(x_1, \dots, x_n, {\overline y})$ be a conjunction of literals
with free variables $x_1, \dots, x_n, {\overline y}$, and 
$\psi({\overline y}) \equiv_{\T} \exists x_1 \dots \exists x_n \phi(x_1, \dots,
  x_n, {\overline y})$. Let $t_1, \dots, t_n$ be $\Sigma'$-terms 
  such that  
$\T' \models \psi({\overline y}) \rightarrow \phi(t_1, \dots t_n,
{\overline y})$. 
Let the substitution $\sigma$ be defined by
$$\sigma(x_i) := \ite{\phi(x_1, \dots, x_n, {\overline y})}{x_i}{t_i},$$ 
and by $\sigma(y) = y$ for the other variables. 
\begin{itemize}
\item[(i)] The substitution $\sigma$ has the following properties:
\begin{itemize}
\item[(a)] $\T' \models \psi({\overline y}) \rightarrow
  \sigma(\phi)$. 
\item[(b)] For any other substitution $\mu$ w.r.t.\ $\Sigma'$ 
such that $\T' \models
  \psi({\overline y}) \rightarrow \mu(\phi)$ we have $\T' \models
  \psi({\overline y}) \rightarrow (\mu(\sigma(z))
  \approx \mu(z))$ for every variable $z \in X$. 
\end{itemize}
\item[(ii)]  Let $\A$ be 
as in assumption (A1), $\A'$ its {\sf ite}-expansion, and $\beta : X
\rightarrow \A'$ be such that $(\A',\beta) \models \psi({\overline
  y})$. Let the set of solutions in $\A'$ w.r.t.\ $\beta$ be: 
$$S_{\A', \beta} = \{ (a_1, \dots, a_n) \in |\A| = |\A'| \mid (\A',
\beta^{\overline a})
\models \phi(x_1,\dots, x_n, {\overline y}) \},$$ 
where $\beta^{\overline a} = \beta[x_1 \mapsto
a, \dots, x_n \mapsto a_n],$ and let 
$$F(a_1, \dots, a_n) = (\A'(\beta^{\overline a})(\sigma(x_1)), \dots,
\A'(\beta^{\overline a})(\sigma(x_n))).$$ 
Then the following hold: 
\begin{enumerate}
\item[(a)] $F(|\A|^n) \subseteq S _{\A', \beta}$, and 
\item[(b)] for every
$(a_1, \dots, a_n) \in S _{\A', \beta}$ we have $F(a_1, \dots, a_n) = (a_1,
\dots, a_n)$.
\end{enumerate} 
Hence, $F(|\A|^n) = S _{\A', \beta}$. 
\end{itemize}
\label{reprod-sol-qe}
\end{thm}
{\em Proof:} 
(i) (a) Let $\A$ be as in assumption (A1).  Let $\A'$
its {\sf ite}-expansion, and $\beta : X \rightarrow \A'$ 
such that $(\A', \beta) \models \psi({\overline y})$. Then 
$(\A', \beta) \models \phi(t_1, \dots, t_n, {\overline y})$ and:

\medskip
$\begin{array}{rl}
\A'(\beta)(\sigma(\phi)) = & \A'(\beta)(\phi(\sigma(x_1),
\dots, \sigma(x_n), {\overline y})) \\
= & \left\{ \begin{array}{ll}
\A'(\beta)(\phi(x_1, \dots, x_n , {\overline y})) & \text{ if }
(\A', \beta) \models \phi(x_1, \dots, x_n, {\overline y}) \\
\A'(\beta)(\phi(t_1, \dots, t_n, {\overline y})) & \text{ otherwise} \end{array}
\right. 
\end{array}$

\smallskip
\noindent so $(\A', \beta) \models \sigma(\phi)$. 
Thus, $\T' \models \forall {\overline y} (\psi({\overline y}) \rightarrow \sigma(\phi))$.

\smallskip
\noindent (i) (b) Let $\mu$ be a substitution w.r.t.\ $\Sigma'$ 
with $\T' \models \forall {\overline y} (\psi \rightarrow \mu(\phi))$. 
Then for all $\beta : X \rightarrow \A'$ with $(\A', \beta) \models
\psi({\overline y})$ we have $(\A', \beta) \models \mu(\phi)$, and therefore:

\smallskip
$\begin{array}{ll}
\A'(\beta)(\mu(\sigma(x_i))) & = \A'(\beta)(\mu({\sf if}~\phi~{\sf
  then}~x_i~{\sf else}~t_i)) \\
& = \A'(\beta)(\ite{\mu(\phi)}{\mu(x_i)}{\mu(t_i)}) =
\A'(\beta)(\mu(x_i)).
\end{array}$
 
\smallskip
\noindent So, for all $x \in X$: 
$\A' \models \psi({\overline y}) {\rightarrow} \mu(\sigma(x)) {\approx} \mu(x)$, so 
$\T' \models \psi({\overline y}) {\rightarrow} \mu(\sigma(x))
{\approx} \mu(x)$. 

\medskip
\noindent (ii) (a) 
Let $(a_1, \dots, a_n) \in |\A|^n = |\A'|^n$. 
From the way $\sigma$ is defined, we have:  
$$
\begin{array}{rl}
\A'(\beta^{\overline a})(\sigma(x_i)) = & 
\A'(\beta^{\overline a})({\sf if}~\phi(x_1, \dots, x_n)~{\sf then}~x_i~
{\sf else}~t_i) \\
 = & \left\{ \begin{array}{ll}
a_i & \text{ if } (\A', \beta^{\overline a}) \models \phi(x_1, \dots,
x_n, {\overline y}) \\
\A'(\beta^{\overline a})(t_i) & \text{ otherwise.} 
\end{array} \right.
\end{array}$$
Thus, if $(\A', \beta^{\overline a}) \models \phi(x_1, \dots,
x_n, {\overline y})$ 
then $F(a_1, \dots, a_n) = (a_1, \dots, a_n) \in S_{\A, \beta}$,
otherwise $F(a_1, \dots, a_n) = (\A'(\beta^{\overline a})(t_1), \dots,
\A'(\beta^{\overline a})(t_n))$, which is also in $S_{\A, \beta}$
because $\T' \models \psi({\overline y}) \rightarrow \phi(t_1, \dots,
t_n, {\overline y})$ and $(\A', \beta)\models \psi({\overline y})$. 

\smallskip
\noindent (ii) (b) If $(a_1, \dots, a_n) \in S _{\A', \beta}$
then $(A', \beta^{\overline a}) \models
\phi(x_1, \dots, x_n, {\overline y})$, 
so $F(a_1, \dots, a_n) = ( \A'(\beta^a)(\sigma(x_1)), \dots, 
\A'(\beta^a)(\sigma(x_n))) = (a_1, \dots, a_n)$. 
\QED

\bigskip
\noindent {\bf Example: Linear arithmetic.}
The following result is an immediate 
consequence of Theorem~\ref{reprod-sol-qe}, and explains how to use 
Theorem~\ref{reprod-sol-qe} and particular solution of constraints 
when $\T = LI({\mathbb R})$ to generate most general solutions. 

Note that the particular solution used in the most general 
solution can contain {\sf if}-{\sf then}-{\sf else} operations.

\begin{thm}
Let $\T = LI({\mathbb R})$ and $\phi(x, y_1, \dots, y_n)$ be a
quantifier-free formula in the signature $\Pi$ of $\T$. 
Assume that $\exists x \phi(x, y_1, \dots, y_n) \equiv
\bigvee_{t \in T} \phi(t, y_1, \dots, y_n)$, where the set of testpoints 
$T = \{ t_1, \dots, t_m \}$ is a set of $\Pi'$ terms.

A conditional most general solution with {\sf if}-{\sf then}-{\sf
  else}-terms (under the condition that 
$\bigvee_{t \in T} \phi(t, y_1, \dots, y_n)$ holds) 
can be expressed by: 

\medskip
\noindent $\begin{array}{l@{}l}
\sigma(x) \,{=}\, & {\sf if}~{\phi(x, x_1, \dots, x_n)}~{\sf then}~x~{\sf
  else}~({\sf if}~\phi(t_1, x_1, \dots, x_n)~{\sf then}~t_1~{\sf else}  \\
& \!\!\!({\sf if}~\phi(t_2, x_1, \dots, x_n)~{\sf then}~t_2~{\sf else} (...
({\sf if}~{\phi(t_m, x_1, \dots, x_n)}~{\sf then}~t_m~{\sf else}~c_f)...)))
\end{array}$   

\medskip
\noindent where $c_f$ is a selected constant which stands for ``no
solution''. 
\label{case-distinction}
\end{thm}
\noindent {\em Proof:} 
$\T = LI({\mathbb R})$ satisfies conditions (A1)--(A3) with ${\cal A}
= {\mathbb R}$. In what follows, let ${\cal A}'$ be the extension of 
${\mathbb R}$ with {\sf if}-{\sf then}-{\sf else} operations. 

We prove that $\sigma$ is a conditional solution: Let $\beta : X
\rightarrow \A'$ with  $(\A', \beta) \models 
\bigvee_{t \in T} \phi(t, y_1, \dots, y_n)$

Let $t_0$ be the {\sf if}-{\sf then}-{\sf else}-term encoding a particular solution obtained from 
$\bigvee_{t \in T} \phi(t, y_1, \dots, y_n)$, if $T = \{ t_1, \dots,
t_m)$ as follows: 
$$\begin{array}{ll}
t_0 = & {\sf if}~\phi(t_1, y_1, \dots, y_n)~{\sf then}~t_1~{\sf else}
\\
&  {\sf if}~\phi(t_2, y_1, \dots, y_n)~{\sf then}~t_2~{\sf else}\\
& \dots \\
&  {\sf if}~\phi(t_m, y_1, \dots, y_n)~{\sf then}~t_m~{\sf else}~c_f
\end{array}$$
Then we have: 

\smallskip
\noindent $\A'(\beta)(\phi(\sigma(x), y_1, \dots, y_n))$

$= \left\{ \begin{array}{ll}
\A'(\beta)(\phi(x, y_1, \dots, y_n)) & \text{ if } (\A', \beta)
\models \phi(x, y_1, \dots, y_n) \\
\A'(\beta)(\phi(t_0, y_1, \dots, y_n)) & \text{ otherwise } 
\end{array} \right.$ 

$= \ite{\A'(\beta)(\phi(x, y_1, \dots, y_n))}{\A'(\beta)\phi(x, y_1, \dots, y_n)}{}$

$\ite{\A'(\beta)(\phi(t_1, y_1, \dots, y_n))}{\A'(\beta)(\phi(t_1, y_1, \dots, y_n))}{}$

$\dots$

$\ite{\A'(\beta)(\phi(t_1, y_1, \dots, y_n))}{\A'(\beta)(\phi(t_m, y_1, \dots, y_n))}{\perp}.$

\medskip
\noindent Since $(\A', \beta) \models \bigvee_{t \in T} \phi(t, x_1,
\dots, x_n)$, it follows that 
$\A'(\beta)(\phi(\sigma(x), y_1, \dots, y_n))$ is true. 
Thus, $\T' \models \bigvee_{t \in T} \phi(t, y_1, \dots, y_n)
\rightarrow \phi(\sigma(x), y_1, \dots, y_n)$. 

\smallskip
\noindent By Theorem~\ref{reprod-sol-qe}, 
$\sigma$ is the most general conditional solution. 
We show how the proof of Theorem~\ref{reprod-sol-qe} specializes 
to the particular situation in which $\T = LI({\mathbb R})$.

Let $\mu$ be another substitution (possibly containing {\sf if}-{\sf
  then}-{\sf else} terms) with 
$$\T' \models \bigvee_{t \in T} \phi(t, y_1, \dots,
y_n)  \rightarrow \phi(\mu(x), y_1, \dots, y_n).$$ 
Let $\beta : X \rightarrow \A'$. 
Assume $(\A', \beta) \models \bigvee_{t \in T} \phi(t, y_1, \dots,
y_n)$. Then $(\A', \beta) \models \phi(\mu(x), y_1, \dots, y_n)$. 
Then, we have: 

\smallskip
$\begin{array}{ll}
\A'(\beta)(\mu(\sigma(x))) := & \ite{ \A'(\beta)(\phi(\mu(x), y_1, \dots, y_n))}{\A'(\beta)(\mu(x))}{} \\
& \ite{\A'(\beta)(\phi(t_1, y_1, \dots, y_n))}{\A'(\beta)(t_1)}{}  \\
& \ite{\A'(\beta)(\phi(t_2, y_1, \dots, y_n))}{\A'(\beta)(t_2)}{}  \\
& \dots \\
& \ite{\A'(\beta)(\phi(t_m, y_1, \dots, y_n))}{\A'(\beta)(t_m)}{\bot}
\end{array}$   

$= \A'(\beta)(\mu(x)).$

\smallskip
\noindent 
Thus, $\T' \models \bigvee_{t \in T} \phi(t, y_1, \dots,
y_n)  \rightarrow (\mu(\sigma(x)) \approx \mu(x))$ 
for every variable $x$.
 \QED

\begin{ex} 
Let $\T = LI({\mathbb R})$ and $\phi(x, x_1, x_2) = (x \geq x_1) \wedge (x \geq x_2)$. 
It is clear that $\exists x \phi(x, x_1, x_2)$ is true. 
To find a solution we proceed as follows: The set of
testpoints is $T = \{ - \infty, x_1, x_2\}$. 
It can be easily seen that the testpoint $- \infty$ is not needed (it cannot be a solution), so  
$$\exists x \phi(x, x_1, x_2) \equiv_{\T} \phi(x_1, x_1, x_2) \vee
  \phi(x_2, x_1, x_2) \equiv_{\T} x_1 \geq x_2 \vee x_2 \geq x_1.$$ 

\noindent Note that $\T \models \forall x_1, x_2 (x_1 \geq x_2 \vee
x_2 \geq x_1)$, so $\exists x \phi(x, x_1, x_2)$ has always a
solution. A solution for $\exists x \phi(x, x_1, x_2)$ is the
{\sf if}-{\sf then}-{\sf else}-term: 

\smallskip
$\begin{array}{rcl} 
t & = & \ite{\phi(x_1, x_1, x_2)}{x_1}{(\!\ite{\phi(x_2, x_1, x_2)}{x_2}{x_1)}}\\
& = & \ite{(x_1 \geq x_1 \wedge x_1 \geq x_2)}{x_1}{(\!\ite{(x_2 \geq x_1
  \wedge x_2 \geq x_2)}{x_2}{x_1)}}\\
& = & \ite{x_1 \geq x_2}{x_1}{x_2}.
\end{array}$

\smallskip
\noindent The substitution $\sigma$ with 
$\sigma(x)= (\ite{\phi(x, x_1, x_2)}{x}{t}~)$, $\sigma(x_1) = x_1$ and
$\sigma(x_2) = x_2$ is the most general solution.  \hfill $\blacksquare$ 
\label{ex-x1-x2}
\end{ex}
\begin{ex}
Let $\T = LI({\mathbb R})$ and $\phi(x, y, z) = (x + y + z \geq 3 \wedge y + 2z \leq 2)$. 
There is a value $x \in {\mathbb R}$ such that $\phi(x, y, z)$ holds 
iff $\exists x (x + y + z \geq 3 \wedge y + 2z \leq 2)$ holds. 
The set of testpoints is in this case $T = \{ - \infty, 3 - y - z \}$.
It can again be seen that the testpoint $- \infty$ is not needed 
(it cannot be a solution), and  
$$\exists x (x + y + z \geq 3 \wedge y + 2z \leq 2) \equiv_{\T} y + 2z \leq
2.$$ 
\noindent If $y + 2z \leq 2$, a solution for $\exists x \phi$ is 
$t = 3 - y - z$. 
By Theorems~\ref{reprod-sol-qe} and~\ref{case-distinction}, the most general (conditional) 
solution, under the assumption that $y + 2z \leq 2$ is: 

\smallskip
$\sigma(x) = \ite{x + y + z \geq 3 \wedge y + 2z \leq 2}{x}{3 - y -
  z}$.

\smallskip
\noindent This means: if $\beta(x) = a, \beta(y) = b, \beta(z) = c \in
{\mathbb R}$ with $b + 2c \leq 2$ then  
we have: 

\smallskip
$({\mathbb R}, \beta) \models y + 2z \leq 2 \quad \text{ and } \quad {\mathbb R}(\beta)(\sigma(x)) = \left\{ \begin{array}{ll}
a & \text{ if } a + b + c \geq 3 \\
3 - b - c & \text{ if } a + b + c < 3
\end{array} \right.$
\hfill $\blacksquare$ 
\label{ex-x-y-z}
\end{ex}
In the next example we illustrate the difference between 
(most general, conditional)
solutions obtained by successively eliminating variables and 
most general conditional solutions obtained by applying
Theorem~\ref{reprod-sol-qe}.
\begin{ex}
Let $\T = LI({\mathbb R})$ and $\phi(x, y, z) = (x + y + z = 3
\wedge y + 2z = 3)$. We analyze two different ways of 
regarding the computation of a solution (resp.\ most general solution)  
for $\exists z \exists y \exists x \phi(x, y, z)$. 

\smallskip
\noindent {\bf Eliminating the variables successively:}
We first analyze the formula 
$$\exists x \phi(x, y, z) \equiv_{\T} y + 2z = 3.$$ 
The set of testpoints is $T = \{ 3 - y - z \}$ ($t_x = 3 - y - z$ is
also a solution for $\phi$ under the condition that $y + 2z = 3$ holds). 
Assume that $y + 2z = 3$. 
The most general solution for $\exists x \phi(x, y, z)$
under the condition $y + 2z = 3$ is, by 
Theorems~\ref{reprod-sol-qe}, the substitution $\sigma_x$ with: 
$$\sigma_x(x) = {\sf if}~ (x + y + z = 3
\wedge y + 2z = 3)~{\sf then}~x~{\sf else}~(3 - y - z) = (3 - y -
z).$$
and $\sigma_x(y) = y, \sigma_x(z) = z$. 
Under the assumption that $y + 2z = 3$ it can be seen that in 
$\A'$ $\sigma_x(x) = (3 - y - z) =  z.$ In other words, 
for every $\beta : X \rightarrow \A'$ with $(\A', \beta) \models y +
2z = 3$
we have $(\A', \beta) \models \sigma_x(x) =  z.$

We next analyze the formula 
$\exists y \exists x \phi(x, y, z) \equiv_{\T} \exists y (y + 2z = 3)$. 
We know that $\exists y (y + 2z = 3) \equiv_{\T} \top$. The solution
(and the unique testpoint) is $t_y = 3 - 2z$. 
The most general solution for $\exists y (y + 2z = 3)$ is the substitution $\sigma_y$ with: 
$$\sigma_y(y) = {\sf if}~ (y + 2 z = 3)~{\sf then}~y~{\sf else}~(3 - 2z) = (3 - 2z).$$
and $\sigma_y(x) = x, \sigma_y(z) = z$. 

Since $\exists z \exists y \exists x \phi(x, y, z) \equiv_{\T} \exists z
\top \equiv_{\T} \top$ any
real number or term $t_z$ is a solution. The most general solution is 
$\sigma_z$ with $\sigma_z(z) = {\sf if}~ \top~ {\sf then}~z~{\sf
else}~t_z = z$, $\sigma_z(u) = u$ for every other variable, i.e.\ the
identity.

\medskip
\noindent {\bf Applying Theorem~\ref{reprod-sol-qe}.}
We now analyze the formula 
$$\exists x, y, z \phi(x, y, z) = \exists x,
y, z (x + y + z = 3 \wedge y + 2z = 3) \equiv_{\T} \top.$$ 
A solution is $x = 1, y = 1, z = 1$. 
Theorem~\ref{reprod-sol-qe} yields the most general solution 
$\sigma$ with 
$$\begin{array}{rcl}
\sigma(x) & = & {\sf if}~ (x + y + z = 3
\wedge y + 2z = 3)~{\sf then}~x~{\sf else}~1 \\
\sigma(y) & = & {\sf if}~ (x + y + z = 3
\wedge y + 2z = 3)~{\sf then}~y~{\sf else}~1 \\
\sigma(z) & = & {\sf if}~ (x + y + z = 3
\wedge y + 2z = 3)~{\sf then}~z~{\sf else}~1. 
\end{array}$$
It can be seen that 
$\T' \models \sigma(x) = \sigma(z)$ and $\T' \models \sigma(y) =
\sigma(3 - 2z)$. 

Indeed, for every $\beta : X \rightarrow {\mathbb R}$ we have: 

$\begin{array}{lcl} 
\A'(\beta)(\sigma(x)) & = & \left\{ \begin{array}{ll}
\beta(x) & \text{ if } \beta(x) + \beta(y) + \beta(z) = 3 \text{ and }
\beta(y) + 2 \beta(z) = 3 \\
1 & \text{ otherwise}
\end{array} \right.\\
& = & \left\{ \begin{array}{ll}
\beta(z) & \text{ if } \beta(x) + \beta(y) + \beta(z) = 3 \text{ and }
\beta(y) + 2 \beta(z) = 3 \\
1 & \text{ otherwise}
\end{array} \right.\\
& = & \A'(\beta)(\sigma(z)).
\end{array}$ 

$\begin{array}{lcl} 
\A'(\beta)(\sigma(y)) & = & \left\{ \begin{array}{ll}
\beta(y) & \text{ if } \beta(x) + \beta(y) + \beta(z) = 3 \text{ and }
\beta(y) + 2 \beta(z) = 3 \\
1 & \text{ otherwise}
\end{array} \right.\\
& = & \left\{ \begin{array}{ll}
3 - 2 \beta(z) & \text{ if } \beta(x) + \beta(y) + \beta(z) = 3 \text{ and }
\beta(y) + 2 \beta(z) = 3 \\
3 - 2 * 1 & \text{ otherwise}
\end{array} \right.\\
& = & \A'(\beta)(\sigma(3 - 2z)).
\end{array}$ 

\medskip
\noindent {\bf Comparison between the two approaches.} 
We compare the substitutions $\sigma_z \circ \sigma_y \sigma_x =
\sigma_y \circ \sigma_x$ and
$\sigma$.  

It is easy to see that $\sigma_z \circ \sigma_y \circ \sigma_x =
\sigma_y \circ \sigma_x$ is a solution for $\exists z \exists y
\exists x \phi(x, y, z)$ under the condition that $y + 2z = 3$: 
Indeed, 
$$\begin{array}{lcl} 
\sigma_y (\sigma_x(\phi(x, y, z))) & = & \sigma_y (\phi(\sigma_x(x),
y, z)) \\
& = & \sigma_y(\phi(3 - y - z, y, z)) \\
& = & \phi(3 - \sigma_y(y) - z, \sigma_y(y), z) \\
& = & \phi(3 - 3 + 2z - z, 3 - 2z, z) = \phi(z, 3-2z, z) \\
& = & ((z + (3-2z)
+ z = 3) \wedge (3-2z)+2z= 3)) \\
& = & ((3 = 3) \wedge (3 = 3)) \equiv \top.
\end{array}$$
By Theorem~\ref{reprod-sol-qe}, $\sigma$ is more general than   
$\sigma_z \circ \sigma_y \circ \sigma_x =
\sigma_y \circ \sigma_x$. Indeed, $\sigma_y(\sigma_x(\sigma(u))) =
\sigma_y(\sigma_x(u))$ for every variable $u$. 

\smallskip
We now prove that $\sigma(\sigma_y(\sigma_x(u))) = \sigma(u)$ for
every variable $u \in X$: 

$$\begin{array}{lcl} 
\sigma(\sigma_y (\sigma_x(x)) & = & \sigma(\sigma_y(z)) = \sigma(z) = \sigma(x)\\
\sigma(\sigma_y(\sigma_x(y)) & = & \sigma(\sigma_y(y)) = \sigma(3  -
2z) = \sigma(y) \\
\sigma(\sigma_y(\sigma_x(z)) & = & \sigma(z)
\end{array}$$
\label{ex-composition}
\end{ex}
Example~\ref{ex-composition} suggests that we can obtain 
most general solutions either by successive variable elimination 
or by applying Theorem~\ref{reprod-sol-qe} directly. 
However, in general, the assumptions and the particular terms 
used for successive variable 
elimination might be different from the assumptions used when 
eliminating all variables at the same time. In
Example~\ref{ex-composition} the particular solution $t_x(y, z)$ for $x$ 
used to compute $\sigma_x(x)$ is a ``uniform'' solution which is a 
solution for all $y, z$, which exists under the condition that 
$y + 2z = 3$. On the other hand, the solution $x = 1, y = 1, z = 1$ 
used when applying Theorem ~\ref{reprod-sol-qe} directly is a 
special solution, which exists without any additional conditions.
In future work we would like to
further investigate the link between these two possible ways of
computing most general (conditional) solutions.

\medskip
\noindent {\bf Example: Boolean algebras.} For $\T = {\sf Bool}$ the results specialize (with the
remarks on the choice of $\T' = {\sf Bool}$ in
Example~\ref{ex:a3}) to Theorem~\ref{discr-constants}.

\section{Solutions depending on given variables}
\label{sect:some-vars}

Consider the problem $\exists x \phi(x, y_1, \dots, y_n)$. We analyze
possibilities of deciding whether there exist solutions which do not
depend on a given variable $y_i$, in situations in which the theory
$\T$ {\em allows quantifier elimination}. 

This type of problem is a generalization of unification with linear
constant restrictions (cf.\ \cite{BaaderSchulz96} and the link between
$E$-unification with linear constant restrictions and validity of positive
sentences mentioned in Theorem~\ref{proposition-baader-positive}). 

\begin{defi}
Let $\T$ be a theory and $\phi(x, y_1, \dots, y_n)$ be a quantifier-free formula
with free variables $x, y_1, \dots, y_n$. 
A term  $t(y_{i_1}, \dots,
y_{i_k})$ is a {\em solution}  for 
$\exists x \phi(x, y_1, \dots,
y_n)$ which only contains 
variables in a subset $\{ y_{i_1}, \dots, y_{i_k} \} \subseteq \{ y_1, \dots,
y_n \}$ iff $\T \models \forall y_1, \dots, y_n \phi(t(y_{i_1}, \dots,
y_{i_k}), y_1, \dots, y_n)$. 
\end{defi}

\begin{defi}
Let $\T$ be a theory allowing quantifier elimination 
and $\phi(x, y_1, \dots, y_n)$ be a quantifier-free formula
with variables $x, y_1, \dots, y_n$. 

A term  $t(y_{i_1}, \dots,
y_{i_k})$ is a {\em conditional solution} for 
$\exists x \phi(x, y_1, \dots,
y_n)$ which only contains 
variables in a subset $\{ y_{i_1}, \dots, y_{i_k} \} \subseteq \{ y_1, \dots,
y_n \}$ if and only if  $\T \models \forall y_1, \dots, y_n (\psi(y_1, \dots, y_n)
\rightarrow \phi(t(y_{i_1}, \dots,
y_{i_k}), y_1, \dots, y_n))$, where $\psi$ is the quantifier-free formula with 
$\psi(y_1, \dots, y_n) \equiv_{\T} \exists x \phi(x, y_1, \dots, y_n)$.
\end{defi}

\begin{thm}
Let $\T$ be a theory allowing quantifier elimination, 
and $\phi(x, y_1, \dots, y_n)$ be a quantifier-free formula
with variables $x, y_1, \dots, y_n$. 
Then we can decide whether there 
is a solution $x = t(y_{i_1}, \dots, y_{i_k})$ for 
$\exists x \phi(x, y_1, \dots, y_n)$ 
which only contains 
variables in a subset $\{ y_{i_1}, \dots, y_{i_k} \} \subseteq \{ y_1,
\dots, y_n \}$. 
\label{thm-solution-vars}
\end{thm}
{\em Proof:} Let $\psi$ be a quantifier-free formula with 
$\psi(y_1, \dots, y_n) \equiv_{\T} \exists x \phi(x, y_1, \dots, y_n)$.
If ${\overline y}$ is the sequence of variables in 
$\{ y_1, \dots, y_n \} \backslash \{ y_{i_1}, \dots, y_{i_k} \}$, 
let \\
 $\xi(x', y_{i_1}, \dots, y_{i_k}) \equiv_{\T} \forall {\overline y}
(\psi(y_1, \dots, y_n) \rightarrow \phi(x', y_1, \dots, y_n))$, where
$x'$ is a new variable. 
Then the following are equivalent:
\begin{itemize}
\item[(i)] There exists a term $t$ which only contains 
the variables $\{ y_{i_1}, \dots, y_{i_k} \}$ such that 
$\T \models \forall {\overline y} (\psi(y_1, \dots, y_n) \rightarrow
\phi(t, y_1, \dots, y_n))$. 
\item[(ii)] There exists a term $t$ which only contains 
the variables $\{ y_{i_1}, \dots, y_{i_k} \}$ such that $\T \models \xi(t, y_{i_1}, \dots, y_{i_k})$. \QED
\end{itemize}

\begin{ex}[Example~\ref{ex-x1-x2} ctd.]
Let $\T = LI({\mathbb R})$. 
A most general solution for $\exists x (x {\geq} x_1 \wedge x {\geq} x_2)$
is the substitution which is the identity for all variables $y \neq x$
and with 

\smallskip
$\begin{array}{ll}
\sigma(x)  =  {\sf if}~x \geq x_1 \wedge x \geq x_2~{\sf then}~x~{\sf
  else}~& ({\sf if}~x_1 \geq x_2~ {\sf then}~x_1~{\sf else}~\\
& ({\sf if}~x_2 \geq x_1~{\sf then}~x_2~{\sf else}~c_f)) 
\end{array}$

\medskip
\noindent We can test whether there is a solution which does not depend on $x_1$
as follows. Note that $\exists x (x \geq x_1 \wedge x \geq x_2) \equiv
\top$, and 
$\forall x_1 (\top \rightarrow x' \geq x_1 \wedge x' \geq x_2) \equiv
\bot,$
so there is no solution which does not depend on $x_1$. 
\hfill $\blacksquare$ \end{ex}

\begin{ex}[Example~\ref{ex-x-y-z} ctd.]
Let $\phi(x, y, z) = (x + y + z \,{\geq}\, 3 \wedge y + 2z \,{\leq}\, 2)$. 
There is a value $x \in {\mathbb R}$ such that $\phi(x, y, z)$ holds 
iff $\exists x \phi(x, y, z)$ holds. We know that $\exists x \phi(x, y, z) =
\exists x (x + y + z \geq 3 \wedge y + 2z \leq 2) \equiv y + 2z \leq
2$. Assume that $y + 2z \leq 2$. 
\noindent To check if there is a solution $x = t(z)$ for $\phi$ 
which does not depend on $y$ we proceed as follows: 
We use quantifier elimination and obtain: 

\smallskip 
$\forall y (y + 2z \leq 2 ~\rightarrow~ x' + y + z \geq 3 \wedge y + 2z
\leq 2)$ 

$\equiv \neg [\exists y (y \leq 2 - 2z \wedge x' + y + z < 3)]  \equiv \neg  \top \equiv \bot$

\smallskip
\noindent so there is no solution for $x$ which does not depend on $y$. 

\medskip
\noindent Consider now the formula $\phi(x, y, z) = (x+y+z \leq 3 \wedge y + 2 z
\leq 2)$. There exists a value $x \in {\mathbb R}$ such that $\phi(x,
y, z)$ holds iff $\exists x \phi(x, y, z)$ holds. 
Also in this case we have $\exists x \phi(x, y, z)  \equiv y + 2z \leq 2$. 
\noindent To check if there is a solution $x = t(z)$ for $\phi$ 
which does not depend on $y$ (under the assumption that 
$y + 2z \leq 2$) we proceed as follows: 
We use quantifier elimination and obtain: 

\smallskip
\noindent $\forall y (y + 2z \leq 2 ~\rightarrow~ x' + y + z \leq 3 \wedge y + 2z
\leq 2)$ 

$\equiv \neg [\exists y (y \leq 2 - 2z \wedge x' + y + z > 3)] \equiv \neg (3 - x' - z < 2 - 2z) \equiv x' \leq z + 1.$

\smallskip
\noindent One can see for instance that -- under the assumption that
$y + 2z \leq 2$ -- we have a solution depending only 
on $z$ (e.g.\ $t = z + 1$ or $t = z$ are such solutions). 
\hfill $\blacksquare$ \end{ex}
The theory of Boolean algebras does not allow in general quantifier
elimination for formulae containing negative literals, so
Theorem~\ref{thm-solution-vars} cannot be applied. 
We present a situation in which
the elimination of the universal quantifiers is possible, 
and Theorem~\ref{thm-solution-vars} can be applied.

\begin{ex} Let $\T = {\sf Bool}$. 
Consider the problem $\exists x ~(x_1 \wedge x) \vee (x_2
\wedge \neg x) \approx 0$.
We know that  $\exists x ~(x_1 \wedge x) \vee (x_2
\wedge \neg x) \approx 0 \equiv_{\sf Bool} x_1 \wedge x_2 \approx 0$. 
 
\noindent We check whether there is a solution which does not depend on $x_2$ as
follows: 

\noindent Note first that 
$(x_1 \wedge x_2 \approx 0) \equiv_{\sf Bool} (x_1 \leq \neg
x_2) \equiv_{\sf Bool} (x_2 \leq \neg x_1)$. 

\noindent Therefore, the following are equivalent in the theory of Boolean algebras:
\begin{enumerate}[(1)]
\item There exists a term $t$ not depending on $x_2$ such that 
$(x_1 \wedge t) \vee (x_2 \wedge \neg t) \approx 0$;
\item ${\sf Bool} \models \forall x_2 (x_1 \wedge x_2 \approx 0 \rightarrow (x_1 \wedge t) \vee
  (x_2 \wedge \neg t) \approx 0)$
\item  ${\sf Bool} \models \forall x_2 (x_2 \leq \neg x_1 \rightarrow (x_1 \wedge t) \vee
  (x_2 \wedge \neg t) \approx 0)$
\item  ${\sf Bool} \models  (x_1 \wedge t) \vee (\neg x_1 \wedge \neg
  t) \approx 0$, i.e.\ ${\sf Bool} \models (t \leq \neg x_1) \wedge (\neg t \leq x_1)$
\item ${\sf Bool} \models  t \approx \neg x_1$.
\end{enumerate}
\noindent 
The equivalence of (1) and (2) follows from the definition.
The equivalence of (2) and (3) is an easy consequence of the fact that in any Boolean
algebra $x_1 \wedge x_2 = 0$ iff $x_2 \leq \neg x_1$.
To prove that (3) implies (4) assume that 
$${\sf Bool} \models \forall x_2 (x_2 \leq \neg x_1 \rightarrow (x_1 \wedge t) \vee
  (x_2 \wedge \neg t) \approx 0).$$
Then, for $x_2 = \neg x_1$, we have $x_2 \leq \neg x_1$, so 
$${\sf Bool} \models  (x_1 \wedge t) \vee (\neg x_1 \wedge \neg
  t) \approx 0$$
i.e.\ ${\sf Bool} \models  (x_1 \wedge t) \approx 0$ and ${\sf Bool}
\models  (\neg x_1 \wedge \neg t) \approx 0$, so: 
$${\sf Bool} \models (t \leq \neg x_1) \wedge  (\neg t \leq x_1).$$
Conversely, assume that (4) holds. 
Then 
$${\sf Bool} \models (x_2 \leq \neg x_1) \rightarrow ((x_1 \wedge t) \vee (x_2 \wedge \neg
  t) \leq (x_1 \wedge t) \vee (\neg x_1 \wedge \neg
  t)$$
so ${\sf Bool} \models (x_2 \leq \neg x_1) \rightarrow ((x_1 \wedge t) \vee (x_2 \wedge \neg
  t) \approx 0)$.
The equivalence of (4) and (5) is immediate.

\smallskip
This means that the only solution which does depend on $x_2$ is $t =
\neg x_1$. 
 
Similarly it can be shown that the only term $t$ which does not depend
on $x_1$ such that $(x_1 \wedge t) \vee (x_2 \wedge \neg t) \approx 0$,
under the assumption that such a term exists (i.e.\ that $x_1 \wedge x_2 =
0$) is $t = x_2$. 
\end{ex}

\section{Second-order quantifier elimination}
\label{sect:soqe}

We now analyze how the idea introduced before can 
be extended to the analysis of second-order problems 
of the form 
$$\exists f  \, \phi(x_1, \dots, x_n),$$ 
where $\phi$ is a quantifier-free formula in the variables $x_1, \dots,
x_n$ in a signature containing the function symbol $f$ w.r.t.\ a
theory $\T$. 

\smallskip
We restrict to extensions $\T = \T_0 \cup {\sf
  UIF}_{\Sigma}$ of a theory $\T_0$ (with signature $\Pi_0$) allowing
quantifier elimination, 
with uninterpreted function symbols in a set $\Sigma$.
The theory $\T_0 \cup {\sf  UIF}_{\Sigma}$ does not in general allow 
quantifier elimination or second-order quantifier elimination. 
In \cite{Sofronie-cade2025,Sofronie-cade2025-arxiv}, a method for eliminating function symbols
  with arity $\geq$ 1 was proposed for sets $G$ 
of flat ground formulae w.r.t.\ a certain class of theory
extensions (local theory extensions, cf.\ \cite{Sofronie-cade-05}). 
In Algorithm~\ref{alg-symb-elim} we present 
a specialization of the method proposed in \cite{Sofronie-cade2025}
to extensions with uninterpreted function symbols.

\begin{algorithm}[t]
\caption{Algorithm for Function Elimination \cite{Sofronie-cade2025,Sofronie-cade2025-arxiv}}

{\small \begin{tabular}{ll}
& \\[-2ex]
{\bf Input:} & $\T := {\cal
  T}_0 \cup {\sf UIF}_{\Sigma}$ where ${\cal T}_0$ has signature
$\Pi_0$ and $\Sigma = \Sigma_s \cup \Sigma_e$ with $\Sigma \cap \Pi_0
= \emptyset$; \\
& $G$, a finite set of flat $(\Pi_0 \cup \Sigma)$-clauses with
variables ${\overline x}$; \\

{\bf Output:} & Quantifier-free $\Pi_s$-formula $\Gamma$ with
variables ${\overline x}$. \\[2ex]
\hline 
\end{tabular}
} 

{\small 
\begin{description}
\vspace{-1mm}
\item[Step 1] Compute from $G$ the set of $\Pi_0$-clauses $G_0 {\cup}
  {\sf Con}_0$ by introducing, in a bottom-up manner, new  
variables $x_t \in C$ for subterms $t {=} g(x_1, \dots, x_k)$ 
where $g {\in} \Sigma$, and storing the 
definitions $x_t {\approx} g(x_1, \dots, x_k)$ in a set ${\sf Def}$, 
and: 

\smallskip
${\sf Con}_0 = \{ \bigwedge_{i = 1}^n x_i \approx y_i \rightarrow x
\approx y \mid 
f(x_1, \dots, x_n) \approx x, 
f(y_1, \dots, y_n) \approx y \in {\sf Def} \}$

\smallskip
\item[Step 2]  Let $G_1 := G_0\cup {\sf Con}_0$. 
Among the variables in $G_1$, identify 
\begin{enumerate}
\item[(i)] the sets ${\overline x}_e$ and resp.\ ${\overline x}_s$ 
consisting of variables $x_f$ 
introduced by a definition 
$x_f {:=} f(x_1,\dots,x_k)$ with $f \in \Sigma_e$ resp.\ $f \in \Sigma_s$; 
\item[(ii)] the remaining variables ${\overline x}$. 
\end{enumerate}
Consider the formula
$\exists {\overline x}_e G_1({\overline x}, {\overline x}_s, {\overline x}_e)$.

\smallskip
\item[Step 3] Compute a quantifier-free formula 
$\Gamma_1({\overline x}, {\overline x}_s)$  equivalent to 
$\exists {\overline x}_e G_1({\overline x}, {\overline x}_s,{\overline
  x}_e)$ w.r.t.\ $\T_0$ using a  method for quantifier elimination in 
${\mathcal T}_0$.  

\smallskip
\item[Step 4] Let $\Gamma({\overline x})$ be the formula 
obtained by replacing back in 
$\Gamma_1({\overline x}, {\overline x}_s)$ 
the variables $x_g$ in ${\overline x}_s$ 
introduced by definitions $x_g := g(x_1, \dots,
x_k)$ with the terms $g(x_1, \dots,x_k)$.\!\!\!\! 

\smallskip
\end{description}
} 
\label{alg-symb-elim}
\end{algorithm}

We identify a situation in which a most general solution w.r.t.\ $\T =
\T_0 \cup {\sf UIF}_{\Sigma}$ can be found. 
We assume that $\T_0$ satisfies assumptions $(A1), (A2)$ and
$(A3)$. 

\smallskip
In what follows we assume, w.l.o.g., that $\phi = G$, where $G$ is a
conjunction of literals. Here we will only consider the case in which 
$G$ is flat (i.e.\ the only arguments of function symbols with arity 
$\geq 1$ are variables) and Algorithm~1 can be applied. 
The case in which $G$ is not flat is briefly discussed at the end of
this section.

\begin{thm}
Let $\T = \T_0 \cup {\sf UIF}_{\Sigma}$ be a theory with signature
$\Pi$, where $\T_0$ allows
quantifier elimination, and let $G$ be a set of flat $\Pi$-clauses. 
Algorithm~\ref{alg-symb-elim} applied to the case when 
$\Sigma_e = \{ f \}$ and $\Sigma_s = \Sigma \backslash \{ f \}$, 
returns a formula $\Gamma({\overline x})$ equivalent to $\exists f G$ w.r.t.\ $\T_0 \cup {\sf
  UIF}_{\Sigma}$. 
\label{soqe}
\end{thm}
{\em Proof:} 
We prove, with the notation in Algorithm~1, that the following two statements hold for every set $G$ of
flat ground clauses:
\begin{itemize}
\item[(i)] For every $\Sigma$-structure $\B$ which is a model of
  $\T_0 \cup {\sf UIF}_{\Sigma}$, and every $\beta : X \rightarrow \B$
  with $(\B, \beta) \models G$ we have 
$(\B, \beta) \models \Gamma({\overline x})$. 
\item[(ii)] For every $\Pi_0 \cup \Sigma_s$-structure $\B$
which is a model of $\T_0 \cup {\sf UIF}_{\Sigma_s}$ and every $\beta
: X \rightarrow \B$ such that 
$(\B, \beta) \models  \Gamma({\overline x})$  there exists a 
$\Sigma$-structure $\C$ such that
$\C_{|\Pi_0 \cup \Sigma_s} = \B_{|\Pi_0 \cup \Sigma_s}$
and $(\C, \beta) \models G$.
\end{itemize}
(i) Let $\B$ be a $\Sigma$-structure which is a model of $\T_0$, and 
let $\beta :  X\rightarrow \B$ such that $(B, \beta) \models G$. 
Let ${\overline \beta}$ be the extension of $\beta$ to the 
new variables ${\overline x}_s, {\overline x}_e$
introduced in {\bf Step 1} of Algorithm 1, defined such that 
 ${\sf Def}$ holds. Then $(\B, {\overline \beta}) \models
 G_1({\overline x}, {\overline x}_s, {\overline x}_e)$, 
so $(\B_{|\Pi_0}, {\overline \beta}) \models G_1({\overline x},{\overline x}_s, {\overline x}_e)$.
$\B_{|\Pi_0}$ is 
a model of $\T_0$, so since $\T_0$ allows quantifier elimination
it follows that
$(\B_{|\Pi_0}, {\overline \beta}) \models \Gamma_1({\overline x}, {\overline x}_s)$
(which is equivalent to $\exists {\overline x}_e G_1({\overline x}, {\overline x}_s,{\overline
  x}_e)$ w.r.t.\ $\T_0$). 
Since $(\B, \overline{\beta}) \models {\sf Def}$, it follows that 
$(\B, \beta) \models \Gamma({\overline x})$. 

\bigskip
\noindent (ii) Let $\B$ be a  $\Pi_0 \cup \Sigma_s$-structure 
which is a model of  $\T_0 \cup {\sf UIF}_{\Sigma_s}$ 
and let $\beta : X \rightarrow \B$ such
that $(\B, \beta) \models \Gamma({\overline x})$.
Let ${\overline \beta}$ be the extension of $\beta$ to the 
new variables ${\overline x}_s$ introduced in {\bf Step 1},
defined according to the subset ${\sf Def}_s$ of ${\sf Def}$ which
contains the definitions of the new variables in ${\overline x}_s$.
Then $(\B, {\overline \beta}) \models \Gamma_1({\overline  x},{\overline x}_s)$.
Since, in particular, $\B$ is a model of $\T_0$ and 
$\Gamma_1({\overline x}, {\overline x}_s) \equiv_{\T_0}
\exists {\overline x}_e  G_1({\overline x}, {\overline x}_s,{\overline
  x}_e)$, 
there exists another valuation ${\overline \beta}'$ which coincides 
with ${\overline \beta}$ on the variables in ${\overline x}$ and
${\overline x}_s$, but possibly has new interpretations for the
variables  in ${\overline x}_e$, such that $(\B, {\overline \beta}')
\models G_1({\overline x}, {\overline x}_s,{\overline  x}_e)$. 

\smallskip
Note that ${\overline  x}_e$ consists of variables  
resulting from renaming terms of the form $f(x_1, \dots, x_n)$.
Let ${\sf Def}_f$ be the subset of ${\sf Def}$ containing definitions
for the variables in ${\overline  x}_e$.

\smallskip
Let $\C = (|\B|, \{ g_{\B} \}_{g \in \Pi_0 \cup \Sigma_s} \cup \{f_{\C} \})$,
where the definition of symbols in $\Pi_0 \cup \Sigma_s$
is as in $\B$, and $f_{\C}$ is defined as follows: 
$$f_{\C}(a_1, \dots, a_m) = \left\{ \begin{array}{ll}
{\overline \beta}'(x)  & \text{ if  there exist } x_1, \dots, x_m \text{ with } a_i
= \beta'(x_i) \text{ and } \\
& ~~~~~~~~~~~~~~~~~~~~~~~~~~~~x \approx f(x_1, \dots, x_m) \in {\sf Def} \\
a & \text{ otherwise}
\end{array} \right.$$
\noindent where $a$ is a fixed element in $|\B|$. 
Since $(\B, {\overline \beta}') \models G_1({\overline x}, {\overline
  x}_s,{\overline x}_e) = G_0 \wedge {\sf
  Con}_0$, and $G$ was assumed to be flat,
it follows that $f_{\C}$ is well-defined. 

\smallskip
Since the definitions of the functions in $\Sigma_s$ in $\C$ are the
same as in $\B$, we have $(\C, \beta') \models {\sf Def}_s$. 
By the way $f_{\C}$ is defined, $(\C, \beta') \models {\sf Def}_f$.
Thus, $(\C, \beta') \models {\sf Def}$.  

\smallskip
Since $(\B, {\overline \beta}') \models G_1({\overline x}, {\overline
  x}_s,{\overline x}_e)$,  from the way $\C$ is defined it follows that 
$\C$ is a model 
of $\T_0 \cup {\sf UIF}_{\Sigma}$. 
with $\C_{|\Pi_0 \cup \Sigma_s} =
\B_{|\Pi_0 \cup \Sigma_s}$ such that $(\C, \beta) \models G$. \QED

\

\

\noindent We make the following assumptions and use 
the following notation.
\begin{description}
\item[(A)] 
$\T = \T_0 \cup {\sf UIF}_{\Sigma}$, where $\T_0$ 
is a theory with signature $\Pi_0$ 
satisfying assumptions (A1), (A2), (A3) with 
${\cal F}$ being the set of all quantifier-free formulae. 
\item[(B)] 
Let $\A$ as in assumption (A1),  
$\Sigma'_0, \Pi_0' = (\Sigma'_0, {\sf Pred})$, $\A'$, and $\T_0'$ as in
Definition~\ref{defi:tprim}, and let $\Sigma' = \Sigma_0' \cup
\Sigma$, and $\Pi' =  (\Sigma', {\sf Pred})$. 
Let $\T'$ be the theory described by 
$\M = \{ \B \mid \B \text{ is a } \Pi'\text{-structure with } \B_{|\Pi'_0} = \A'
\}$, consisting of all $\Pi$-structures whose reduct to $\Pi'_0$ is $\A'$.

\item[(C)] $G$ is a 
conjunction of flat literals and -- with the notation used in 
Algorithm~1 -- $G_1 = G_0 \cup {\sf Con}_0 \in {\cal F}$,  
$\Sigma_e = \{ f  \}$ and $\Sigma_s = \Sigma \backslash \{ f \}$. 
If there exists a $\Sigma'$-term $t({\overline z})$ with free
variables ${\overline z}$ with $\T' \models \forall {\overline z} (f_1({\overline z}) =
t({\overline z}))$, we denote
by $G[f \mapsto f_1]$ the formula obtained from $G$ by replacing 
all occurrences $f({\overline z})$ with the term $f_1({\overline z}) =
t({\overline z})$.
\end{description}
\begin{lem} Assume that 
$\T$, $G$ and $f$ satisfy conditions (A), (B), (C).  
Let 
${\overline x}_e = x^f_1, \dots, x^f_k$ be the variables
introduced by definitions $x^f_i = f({\overline x}_i)$ in Algorithm~1.
With the notation in Algorithm 1, 
let $\sigma$ be a solution  w.r.t.\ $\T_0$ for the 
problem 
$\exists {\overline x}_e G_1({\overline x}, {\overline x}_s, {\overline x}_e)$ under the condition that
$\Gamma_1$ holds. Assume that, for every $1 \leq i \leq k$,
$\sigma(x^f_i) = t_i({\overline x}, {\overline x}_s)$, where $t_i({\overline x}, {\overline x}_s)$ is
a $\Sigma'$-term containing variables ${\overline x}, {\overline x}_s$.
\begin{enumerate}
\item[(i)] 
For every model $\B$ of $\T'$ in ${\cal M}$ and every valuation 
$\beta$ with $(\B, \beta) \models \Gamma({\overline x})$ 
let $\C$ be obtained from $\B$ by interpreting all  symbols in
$\Pi_0 \cup \Sigma_s$ as in $\B$ and interpreting $f_{\C} = f_0$, where for all 
$a_1, \dots, a_m \in |\A| = |\B| = |\C|$: 

\smallskip
\noindent $f_0(a_1, \dots, a_m) := \left\{ \begin{array}{ll}
\B(\beta)(\sigma(x_f)) & \text{ if for all } i \text{ we have } a_i =
\beta(x_{j_i}) \\
& \text{ with } f(x_{j_1}, \dots, x_{j_m}) \approx x_f \in {\sf Def} \\
c_{\B}  & \text{ otherwise} 
\end{array} \right.$

\smallskip
\noindent where $c$ is a fixed constant in $\Pi_0 \cup \Sigma \cup C$.
Then $(\C, \beta) \models G$. 

\item[(ii)] 
If ${\overline x}_e = x^f_1, \dots, x^f_k$, are the variables
introduced by definitions $x^f_i = f({\overline x}_i)$, let 

\smallskip
$\begin{array}{ll} 
f_0({\overline z}) = & ({\sf if}~{\overline z} \approx {\overline
  x}_1~{\sf then}~\overline{\sigma(x^f_1)}~{\sf else} \\
& ({\sf if}~{\overline z} \approx {\overline
  x}_2~{\sf then}~\overline{\sigma(x^f_2)}~{\sf else} \\
& \dots \\
& {\sf if }~{\overline z} \approx {\overline x}_k~{\sf
  then}~\overline{\sigma(x^f_k)}~{\sf else}~c), 
\end{array}$ 

\smallskip
\noindent where if ${\overline z} = z_1 \dots z_m$ and 
${\overline x_i} = x_{i1} \dots x_{im}$ then 
${\overline z} \approx {\overline x_i}$ is an
abbreviation for $\bigwedge_{j = 1}^n z_j \approx x_{ij}$, 
and $\overline{\sigma(x^f_i)}$ is obtained from the term 
$\sigma(x^f_i)$ by replacing all variables in ${\overline x}_s$ 
with the terms they represent according to ${\sf Def}$.

Then for every model $\B$ of $\T'$ in ${\cal M}$ and every valuation 
$\beta$ with $(\B, \beta) \models \Gamma({\overline x})$ we have 
$(\B, \beta) \models G[f \mapsto f_0]$.
\end{enumerate}
\label{part-sol}
\end{lem}
{\em Proof:} 
(i) Let $\B \in {\cal M}$ and 
$\beta$ be a valuation such that 
$(\B, \beta) \models \Gamma({\overline x})$. 
We can extend $\beta$ to 
the new variables ${\overline x}_s$ introduced in {\bf Step 1},
defined according to the subset ${\sf Def}_s$ of ${\sf Def}$ which
contains the definitions of the new variables in ${\overline x}_s$
and obtain a valuation $\beta'$ such that 
$(\B, \beta') \models \Gamma_1({\overline x}, {\overline x}_s)$.
Since $\sigma$ is a solution for $\exists {\overline x}_e
G_1({\overline x}, {\overline x}_s, {\overline x}_e)$ 
under the condition that $\Gamma_1({\overline x}, {\overline x}_s)$
holds, we know that -- if 
${\overline x}_e = x^f_1, \dots, x^f_k$
where $x^f_i = f({\overline x}_i) \in {\sf Def}$ -- we have:
$$\T' \models \Gamma_1({\overline x}, {\overline x}_s) \rightarrow G_1({\overline x}, {\overline x}_s, 
\sigma(x^f_1), \dots, \sigma(x^f_k)),$$
so $(\B, \beta') \models G_1({\overline x}, {\overline x}_s, 
\sigma(x^f_1), \dots, \sigma(x^f_k))$. 

Let $\C$ be obtained from $\B$ by interpreting all symbols in
$\Pi_0 \cup \Sigma_s$ as in $\B$ and interpreting 
$f_{\C} = f_0$, where for all 
$a_1, \dots, a_m \in |\A| = |\B| = |\C|$: 

\smallskip
\noindent $f_0(a_1, \dots, a_m) := \left\{ \begin{array}{ll}
\B(\beta')(\sigma(x_f)) & \text{ if for all } i \text{ we have } a_i =
\beta(x_{j_i}) \\
& \text{ with } f(x_{j_1}, \dots, x_{j_m}) \approx x_f \in {\sf Def} \\
c_{\B}  & \text{ otherwise} 
\end{array} \right.$

\smallskip
\noindent The function $f_0$ is well-defined because 
$(\B, \beta') \models G_1({\overline x}, {\overline x}_s, 
\sigma(x^f_1), \dots, \sigma(x^f_k))$, and 
$G_1({\overline x}, {\overline x}_s, {\overline x}_e) = G_0 \wedge {\sf
  Con}_0$, hence for every 
$(\bigwedge_{i = 1}^m x_i \approx y_i) \rightarrow x^f \approx y^f \in
{\sf Con}_0$, where 
$x^f \approx f(x_1, \dots, x_m), y^f \approx  f(y_1, \dots, y_m) \in
{\sf Def}$, we have 
$$(\B, \beta') \models (\bigwedge_{i = 1}^m x_i \approx y_i) \rightarrow
\sigma(x^f) \approx \sigma(y^f).$$ 
Then, with the new definition for $f_{\C}$, 
$(\C, \beta') \models G_1({\overline x}, {\overline x}_s, 
f({\overline x}_1), \dots, f({\overline x}_k))$, 
so, since $(C, \beta') \models {\sf Def}_s$, we have 
$(\C, \beta) \models G({\overline x})$.

\smallskip
\noindent (ii) is a consequence of the fact that 
$(\B, \beta) \models G[f \mapsto f_0]$ iff $(\C, \beta) \models G$, 
where $\C$ is the structure 
obtained from $\B$ by changing the interpretation of $f$ to $f_0$. \QED

\begin{thm} Assume that 
$\T$, $G$ and $f$ satisfy conditions (A), (B), (C), and the premises
of Lemma~\ref{part-sol} hold.
Let $\Gamma({\overline x})$ be the formula returned by 
Algorithm~\ref{alg-symb-elim} when used to 
eliminate $f$ from $G$. 
Let $f_0$ be the particular solution in Lemma~\ref{part-sol}. 
Then ${\tilde \sigma}(f) = (\!\!\ite{G}{f}{f_0})$ is a most general solution 
for $\exists f G({\overline x})$ under assumptions $\Gamma$ in the
following sense:
\begin{enumerate}
\item[(i)] For every model $\B$ of $\T'$ in ${\cal M}$ and 
every $\beta : X \rightarrow \B$, if $(\B, \beta) \models 
\Gamma({\overline x})$ then 
$(\B, \beta) \models G[f \mapsto {\tilde \sigma}(f)]$. 

\item[(ii)] Let $\mu(f)$ be a solution for $\exists f \,
  G({\overline x})$ under the condition that $\Gamma$ holds, 
i.e.\ such that 
$\T' \models \Gamma({\overline x})  \rightarrow G[f \mapsto
\mu(f)]({\overline x})$.  

Assume that $\mu(f)$ has the property that for all 
  variables $x_1,\dots,x_m$ among those in $G$, we can express 
  $\mu(f)(x_1,\dots, x_m)$
  as a $\Sigma'$-term $t({\overline x})$ 
  containing only variables in the set 
  ${\overline x}$ of
  variables in $G$.  

  Then $\T' \models \Gamma \rightarrow
(\mu({\tilde \sigma}(f))(x_1, \dots, x_m) \approx \mu(f)(x_1, \dots,
x_m))$ for all such variables $x_1, \dots, x_m$.
\end{enumerate}
\label{soqe-mgu}
\end{thm}
{\em Proof:} 
(i) Let $\B \in {\cal M}$ and $\beta : X \rightarrow \B$ with 
$(\B, \beta) \models \Gamma({\overline x})$. 
It is not necessarily 
the case that $(\B, \beta) \models G({\overline x})$. We show 
that $(\B, \beta) \models G[f \mapsto {\tilde \sigma}(f)]$.

\smallskip
$\B(\beta)(G[f \mapsto {\tilde \sigma}(f)]) =
 \left\{ \begin{array}{ll}
\B(\beta)(G({\overline x})) & \text{ if } (\B, \beta) \models
G({\overline x}) \\
\B(\beta)(G[f \mapsto f_0]({\overline x})) & \text{ otherwise}
\end{array} \right.$ 

\smallskip
\noindent By Lemma~\ref{part-sol}, $(\B, \beta) \models G[f \mapsto
f_0]$, so $(\B, \beta) \models G[f \mapsto {\tilde \sigma}(f)])$.

\medskip
\noindent (ii) Let now $\mu$ be a map which substitutes $f$ 
with an ``{\sf if}-{\sf then}-{\sf else}'' expression $\mu(f)$ such
that $\mu(f)$ is a solution for $\exists f \,
  G({\overline x})$ under the condition that $\Gamma$ holds, 
i.e.\ 
$\T' \models \Gamma({\overline x})  \rightarrow G[f \mapsto
\mu(f)]({\overline x})$.  
Assume that $\mu(f)$ has the property that for all 
  variables $x_1,\dots,x_m$ among those in $G$, we can express 
  $\mu(f)(x_1,\dots, x_m)$
  as a $\Sigma'$-term $t({\overline x})$ 
  containing only variables in the set 
  ${\overline x}$ of
  variables in $G$.  

Let $\B$ be a model of $\T'$ in ${\cal M}$ and $\beta : X \rightarrow
\B$ such that $(\B, \beta) \models \Gamma({\overline x})$. 
Then clearly $(\B, \beta) \models G[f \mapsto
\mu(f)]({\overline x})$, because $\mu$ is a solution under assumptions
$\Gamma$. 
Then, if ${\overline z} = x_1 \dots x_m$ we have:

\smallskip
$\begin{array}{ll}
\B(\beta)(\mu({\tilde \sigma}(f))({\overline z})) & =  \B(\beta)(\mu({\sf if}~(G[f]~{\sf
  then}~f~{\sf else}~f_0))({\overline z})) = \\
& = \B(\beta)(({\sf if}~G[f \mapsto \mu(f)]~{\sf
  then}~\mu(f)~{\sf else}~\mu(f_0))({\overline z})) \\
& = \B(\beta)({\sf if}~G[f \mapsto \mu(f)]~{\sf
  then}~\mu(f)({\overline z})~{\sf else}~f_0({\overline z})) \\
& = \left\{ \begin{array}{ll} 
\B(\beta)(\mu(f)({\overline z})) & \text{ if } (\B, \beta) \models G[f \mapsto
\mu(f)]\\
\B(\beta)(f_0({\overline z}))  & \text{ otherwise }
\end{array} \right.\\
& =  \B(\beta)(\mu(f)({\overline z}))
\end{array}$

\QED

\begin{ex} Let $\T = LI({\mathbb R}) \cup {\sf UIF}_{\{ f, g\}}$ and
  $G = f(x_1) \approx g(x_2) \wedge f(x_1) \leq x_3
  \wedge g(x_2)  \geq x_4$. 
We can use Algorithm~1 to eliminate $f$ as follows: 

\begin{description}
\item[Step 1:] With the definitions ${\sf Def} = \{
  x_f = f(x_1), x_g = g(x_2) \}$ we obtain 

$G_1 = G_0  = x_f \approx x_g \wedge x_f
  \leq x_3 \wedge x_g \geq x_4$. 

No congruence axioms need to be considered.
\item[Step 2:] Only $x_f$ needs to be eliminated. The test point is $T = \{ x_g \}$.
\item[Step 3:] We obtain $\exists x_f (x_f \approx x_g \wedge x_f
  \leq x_3 \wedge x_g \geq x_4) \equiv x_g \leq x_3 \wedge x_g \geq x_4$. 
\item[Step 4:]  We replace back constants and obtain: $\Gamma = g(x_2) 
  \leq x_3 \wedge g(x_2) \geq x_4$.
\end{description}
Under the condition that 
$g(x_2)
  \leq x_3 \wedge g(x_2) \geq x_4$ a solution exists. It can be constructed
  as follows: 

$f_0(x) = \left\{ \begin{array}{ll}
g(x_2) & \text{ if } x \approx x_1 \\
k   & \text{ otherwise} \end{array} \right. = {\sf if}~x \approx x_1~{\sf
    then}~g(x_2)~{\sf else}~c$

\noindent Based on this, we can construct the most general solution: 

${\tilde \sigma}(f)(x) =  {\sf if}~f(x_1) \approx g(x_2) \wedge f(x_1) \leq x_3 \wedge g(x_2)
  \geq x_4~{\sf then}~f(x)~{\sf else}~f_0(x).
$ 
\hfill $\blacksquare$
\end{ex}

\

\noindent {\bf The case in which $G$ is not flat.}
If $G$ is not flat, Algorithm~\ref{alg-symb-elim} cannot be used. 
We can flatten $G$, by introducing new variables. 
This means that the solution $\sigma(f)$ might contain also other
constants in addition to the ones initially contained in $G$. 
We illustrate the situation on an example, where 
$\T = LI({\mathbb R}) \cup {\sf UIF}_{\{ f, g \}}$. 

\begin{ex}
Consider the unification problem $\exists f (g(f(x)) \approx
f(g(x)))$. 
After flattening  $g(f(x)) \approx f(g(x))$ we obtain 
$y \approx f(x) \wedge z \approx g(x) \wedge g(y) \approx f(z)$. 

\smallskip
$\exists f (g(f(x)) \approx f(g(x))) \equiv \exists y \exists z
\exists f (y \approx
f(x) \wedge z \approx g(x) \wedge g(y) \approx f(z)).$

\smallskip
\noindent We can use Algorithm~\ref{alg-symb-elim} to eliminate $f$ in $\exists f  (y \approx
f(x) \wedge z \approx g(x) \wedge g(y) \approx f(z))$ 
as follows:
\begin{description}
\item[Step 1:] We purify and add instances of the congruence axioms.
If we add the definitions ${\sf Def} :=\{ x_f = f(x), x_g = g(x),  y_g 
= g(y), z_f = f(z) \}$ we obtain:

\smallskip
$y \approx x_f \wedge z \approx x_g \wedge y_g \approx z_f \wedge
(x \approx z \rightarrow x_f \approx z_f) \wedge (x \approx y
\rightarrow x_g \approx y_g)$.

\smallskip
\item[Step 2:] Only $f$ is eliminated, so only $x_f, z_f$ are eliminated.
\item[Step 3:] We compute 

$\exists x_f \exists z_f (y \approx x_f \wedge z \approx x_g \wedge y_g \approx z_f \wedge
(x \approx z \rightarrow x_f \approx z_f) \wedge (x \approx y
\rightarrow x_g \approx y_g))$

$\equiv (z \approx x_g \wedge (x \approx z \rightarrow y \approx
y_g) \wedge (x \approx y \rightarrow x_g \approx y_g))$;  $T_{x_f} = \{
y \}, T_{z_f} = \{ y_g \}$.

The unique solution (and hence the most general solution) is the
substitution 
$\sigma_{f}$ with $\sigma_f(x_f) = y$, $\sigma_f(z_f) = y_g$, and
$\sigma_f(u) = u$ for all the other variables.

\item[Step 4:] We replace the new constants with the terms they
  represent and obtain:

$(z \approx g(x) \wedge (x \approx z \rightarrow y \approx g(y)))$.

\smallskip
 $\overline{\sigma_f(x_f)} = y$, $\overline{\sigma_f(z_f)} = g(y)$, and
$\overline{\sigma_f(u)} = u$ for all the other variables.
\end{description}
The initial problem has a solution iff $\exists y \exists z (z \approx g(x) \wedge (x \approx z \rightarrow y \approx g(y)))$ holds.

Note that $\exists y \exists z (z \approx g(x) \wedge (x \approx z \rightarrow y \approx g(y))) \equiv 
\exists y (x \approx g(x) \rightarrow y \approx g(y))$, which holds
for every model of $\T$ (with witness $y = x$).

\smallskip
\noindent The variables $x_f$ and $z_f$ are introduced by the
definitions $x_f = f(x)$ and $z_f = f(z)$. 
By Lemma~\ref{part-sol} we can define a special solution 
$f_0$ by 

$\begin{array}{lcl}
f_0(u) & = & {\sf if}~u {\approx} x~{\sf then}~\overline{\sigma_{x_f}(x_f)}~{\sf
else}~{\sf if}~u {\approx} z~{\sf then}~\overline{\sigma_{z_f}(z_f)}~{\sf else}~c \\
& = & {\sf if}~u {\approx} x~{\sf then}~y~{\sf
else}~{\sf if}~u {\approx} z~{\sf then}~g(y)~{\sf else}~c \\
\end{array}$.  

\smallskip
\noindent Then the most general solution is:

\smallskip
$\begin{array}{lll}
{\tilde \sigma}(f) & = & \ite{(y \approx
f(x) \wedge z \approx g(x) \wedge g(y) \approx f(z))}{f}{f_0} \\
\end{array}$

\smallskip
\noindent Note that $y$ and $z$ occur in $f_0$ and in 
the expression for ${\tilde \sigma}(f)$, so ${\tilde \sigma}(f)$ depends on both $y$
and $z$.

\smallskip
\noindent We conjecture that if we replace $y$ with $f(x)$, 
and $z$ with $g(x)$ in ${\tilde \sigma}(f)$ 
we might obtain a most general solution 
for $\exists f \, G$ (according to conditions similar to those in 
Theorem~\ref{soqe-mgu}) possibly under a suitable condition 
which states that a solution exists. 
An analysis of the problem of determining most general
solutions in the case in which $G$ is not flat in full generality is
planned for future work.
\hfill $\blacksquare$ \end{ex}

\section{Conclusion}
\label{conclusion}

We analyzed possibilities of 
constructing most general solutions of 
formulae of the form $\exists x_1, \dots, \exists x_n \phi(x_1, \dots,
x_n, y_1, \dots, y_m)$ w.r.t.\ certain
theories $\T$ allowing quantifier elimination,  
where $\phi$ is a quantifier-free conjunction of literals 
in the signature of $\T$, and the free variables $y_1, \dots, y_m$ 
are regarded as parameters. 
We proved that if we can extend the language with a type of 
``{\sf if}-{\sf then}-{\sf
  else}''  constructions, we can describe 
the most general solution of such formulae as terms.  
The idea generalizes  
results about the existence of most general unifiers in discriminator
varieties. We then
considered possibilities of generating most general solutions for
certain problems related to second-order elimination in certain
extensions of theories allowing quantifier elimination with free
function symbols.

In future work we would like to extend the results to the situation in
which the extension with {\sf if}-{\sf then}-{\sf else} constructs
$\T'$ is not described by one model only, but by a class of
models, and to better understand the links between
these results and results on unification for varieties generated by
primal algebras, in discriminator varieties and beyond. We would also
like to analyze the applicability of 
the results on generating most general solutions for 
second order existential constraints in the 
context of higher-order unification (cf.\ e.g.\ \cite{SnyderG89,Dowek01}). 

\begin{credits}
\subsubsection*{\ackname} 
We thank the reviewers for their helpful comments. 
The research reported here was funded by the Deutsche
Forschungsgemeinschaft (DFG, German Research Foundation) – grant
465447331.

\end{credits}

\end{document}